\documentclass[runningheads,a4paper]{llncs}

\usepackage{subfig}
\usepackage[usenames,dvipsnames]{xcolor}
\usepackage[colorlinks,citecolor=blue,linkcolor=BrickRed]{hyperref}
	\usepackage{latexsym,graphicx,epsfig,color}
\usepackage{amsfonts,amssymb,amsmath,amstext}
\usepackage{enumerate}
\usepackage{url}
\usepackage{multirow}
\usepackage{rotating}
\usepackage{makeidx}
\usepackage{tikz}
\usepackage{lipsum}

\allowdisplaybreaks

\newcommand{\IGNORE}[1]{}

\usetikzlibrary{decorations.markings}
\usetikzlibrary{arrows}
\tikzstyle{block}=[draw opacity=0.7,line width=1.4cm]
\tikzstyle{graphnode}=[circle, draw, fill=black!20, inner sep=0pt, minimum width=6pt]
\tikzstyle{point}=[circle, draw, fill=black!30, inner sep=0pt, minimum width=1pt]
\tikzstyle{input}=[rectangle, draw, fill=black!75,inner sep=3pt, inner ysep=3pt, minimum width=4pt]
\tikzstyle{unmatched}=[graphnode,fill=black!0]
\tikzstyle{shaded}=[graphnode,fill=black!20]
\tikzstyle{matched}=[graphnode,fill=black!100]  	
\tikzstyle{matching} = [ultra thick]
\tikzset{
    >=stealth',
    pil/.style={
           ->,
           thick,
           shorten <=2pt,
           shorten >=2pt,}
}
\tikzset{->-/.style={decoration={
  markings,
  mark=at position .5 with {\arrow{>}}},postaction={decorate}}}

\makeindex

\newtheorem{observation}[theorem]{Observation}
\newtheorem{fact}[theorem]{Fact}
\newtheorem{myclaim}[theorem]{Claim}

\newcommand{\NP}{$\mathbf{NP}$}

\def\C   {\mathcal{C}}

\newif\ifFULL
\FULLtrue

\allowdisplaybreaks
\usepackage{amssymb}
\usepackage{amsmath}
\usepackage{graphicx}
\usepackage{algorithm,algpseudocode}

\usepackage{url}
\urldef{\mailsa}\path|{ggurugan,ssingla}@cs.cmu.com|    
\newcommand{\keywords}[1]{\par\addvspace\baselineskip
\noindent\keywordname\enspace\ignorespaces#1}

\usepackage{xspace}
\newcommand\thickbar[1]{\bar{#1}}
\newcommand{\thicktilde}[1]{\mathbf{\tilde{\text{$#1$}}}}

\algtext*{EndFor}
\algtext*{EndIf}

\newcommand{\G}{\mathcal{G}}
\newcommand{\A}{\ensuremath{\mathcal{A}}\xspace}

\newcommand{\M}{\mathcal{M}}

\newcommand{\thickz}{$\mathcal{Z}$\xspace}
\newcommand{\E}{\ensuremath{\mathbb{E}}}
\newcommand{\tE}{\ensuremath{\thicktilde{E}}}
\newcommand{\trE}{\ensuremath{\thicktilde{E}_r}}
\newcommand{\tI}{\ensuremath{\thicktilde{I}}}
\newcommand{\rank}{\mathsf{rank}}

\newcommand{\Mi}{\M_i}
\newcommand{\Mni}{\M_{\thickbar{\imath}}}

\newcommand{\Greedy}{\text{\sc{Greedy}}\xspace}
\newcommand{\MkGreedy}{\text{\sc{Marking-Greedy}}\xspace}
\newcommand{\SampAlg}{\text{\sc{Samp-Alg}}\xspace}
\newcommand{\I}{I}
\newcommand{\Ni}{N_i}
\newcommand{\Nni}{N_{\thickbar{\imath}}}

\newcommand{\OPT}{\textsc{OPT}}

\newcommand{\K}{\ensuremath{\Psi}}

\newcommand{\cl}{{\mathsf{span}}}
\newtheorem{invariant}[theorem]{Invariant}
\newtheorem{updates}[theorem]{Updates}

\newcounter{note2}[section]

\begin{document}

\mainmatter  

\title{ Online Matroid Intersection:\\ Beating Half for Random
    Arrival\footnote{Supported in part by NSF awards CCF-1319811,
        CCF-1536002, and CCF-1617790}}

\titlerunning{Online Matroid Intersection}

%
%

\author{ 
    Guru Guruganesh
    \and Sahil Singla
}
	
%
\authorrunning{Guru Guruganesh, Sahil Singla}

\institute{ Computer Science Department,\\ 
    Carnegie Mellon University,\\ 
\mailsa
}

%
%

\toctitle{Online Matroid Intersection: Beating Half for Random Arrival}
\tocauthor{Guru Guruganesh, Sahil Singla}
\maketitle


\begin{abstract}
   For two matroids $\M_1$ and $\M_2$ defined on the same ground set $E$, the
   online matroid intersection problem is to design an algorithm that constructs a large common independent
   set in an online fashion.  The algorithm is presented with the ground set
   elements one-by-one in a uniformly random  order. At each step, the
   algorithm must irrevocably  decide whether to pick the element, while always maintaining a common
   independent set.  While the natural greedy algorithm---pick an element whenever
   possible---is half competitive, nothing better was previously known;
 even for the special case of online bipartite matching in the edge arrival model.
  We present the first randomized online algorithm that has a $\frac12 + \delta$ competitive ratio in   expectation, where $\delta >0$ is a constant.  
The expectation is over the random order and the  coin tosses of the algorithm. 
As a corollary, we also obtain
    the first linear time algorithm that beats half competitiveness for offline matroid intersection.
    
\IGNORE{
   We study the online matroid intersection problem, which is related to the
   well-studied online bipartite matching problem in the vertex arrival model.
   For two matroids $\M_1$ and $\M_2$ defined on the same ground set $E$, the
   problem is to design an algorithm that constructs a large common independent
   set in an online fashion.  The algorithm is presented with the ground set
   elements one-by-one in a uniformly random    order. At each step, the
   algorithm must irrevocably
   decide whether to pick the element, while always maintaining a common
   independent set.  Since the greedy algorithm --- pick the element whenever
   possible --- has a competitive ratio of half, the natural question is
   whether we can beat half. This problem generalizes online bipartite matching
   in the edge arrival model where a random edge is presented at each step; nothing
   better than half-competitiveness was previously known.

    In this paper, we present a simple randomized algorithm for online matroid
    intersection that has a $\frac12 + \delta$ competitive ratio in
    expectation, where $\delta >0$ is a constant.  The expectation is over the randomness of the input order and the
    coin tosses of the algorithm. As a corollary, we obtain
    the first algorithm that beats half-competitiveness in the bipartite matching
    setting.   We also extend our  result to intersection
    of $k$ matroids and to general graphs, cases not captured by intersection
    of two matroids.  
}
    \keywords{online algorithms; matroid intersection; randomized algorithms; competitive analysis; linear-time algorithms}
\end{abstract}


\section{Introduction}

The \emph{online matroid intersection problem} in the random arrival model (OMI)
consists of two matroids $\M_1 = (E,\mathcal{I}_1)$  and $\M_2 =
(E,\mathcal{I}_2)$, where  the elements in $E$ are presented one-by-one to an
online algorithm whose goal is to construct a
large common independent set.  As an element
arrives, the algorithm must immediately and irrevocably decide whether to \emph{pick} it, 
while ensuring that the picked elements always form a common independent set.
We assume that the algorithm knows the size of $E$ and has access to
independence oracles for the already arrived elements. The greedy algorithm,
which picks an element whenever possible, is half-competitive. The following is
the main result of this paper.

\begin{theorem} \label{thm:matroidmain}
    The online matroid intersection problem in the random arrival model has a
    $(\frac{1}{2} + \delta)$-competitive randomized algorithm, where $\delta >0$ is a constant.  
\end{theorem}

Our OMI  algorithm  makes only a linear number of calls to the independence oracles of both the
matroids.  Given recent interest in finding fast approximation algorithms for
fundamental polynomial-time  problems, this result is of independent
interest even in the offline setting.
Previously known algorithms that perform better than the greedy algorithm
construct an ``auxiliary graph'', which already takes quadratic time~\cite{ChekuriQuanrud-SODA16,HKK-SODA16}.

\begin{corollary}\label{thm:offmatrInter}
The matroid intersection problem has a linear time $(\frac12+\delta)$ approximation algorithm, where $\delta >0$ is a constant.  .
\end{corollary}

A special case of OMI where both the matroids are partition matroids already
captures the \emph{online bipartite matching problem} in the random
edge arrival (OBME) model. Here, edges of a fixed (but adversarially chosen) bipartite graph $G$ arrive 
in a uniformly random order and the algorithm must irrevocably decide whether
to pick them into a matching.
 Despite tremendous progress made in the online vertex arrival
model~\cite{KVV-STOC90,MSVV-Jounal07,GM-SODA08,WangWong-ICALP15,KMZ-STOC15},
nothing non-trivial was known in the edge arrival model where the edges arrive
one-by-one.  
We present the first algorithm that
performs better than greedy in the random arrival model.  Besides being a
natural theoretical question, 
it captures various online content systems such as online libraries where the
participants are known to the matching agencies but  the requests arrive in an
online fashion. 


\begin{corollary}\label{thm:bipmatching}
    The  online bipartite matching problem in the random edge arrival model has
     a $(\frac{1}{2} + \delta)$-competitive randomized algorithm, where $\delta >0$ is a constant.
\end{corollary}



Finally, the simplicity of our OMI algorithm  allows us to extend our
results to the much more general problems of online matching in general graphs and to online
$k$-matroid intersection; the latter problem being  \NP-Hard (proofs 
in \ifFULL Section~\ref{sec:generalgraphs} and Section~\ref{sec:kmatInter}, respectively \else  full version\fi).

\begin{theorem} \label{thm:generalgraphs}
The online matching problem for general graphs in the random edge arrival model has
     a $(\frac{1}{2} + \delta')$-competitive randomized algorithm, where $\delta' >0$ is a constant.
\end{theorem}

\begin{theorem} \label{thm:matroid-k-main}
    The online $k$-matroid intersection problem in the random arrival model has a
    $\left(\frac{1}{k} + \frac{\delta''}{k^4} \right)$-competitive randomized algorithm, where $\delta''>0$ is a constant.  
\end{theorem}

\subsection{Comparison to Previous Work}\label{sec:prevworks}
Our main OMI result is interesting in two different aspects:  It
gives the first linear time algorithm that beats greedy for the classical
offline matroid intersection problem; also, it is the first non-trivial
algorithm for the general problem of online matroid intersection, where
previously  nothing better than half was known even for online bipartite
matching. Since offline matroid intersection problem is a fundamental problem
in the field of combinatorial optimization~\cite[Chapter 41]{Schrijver-Book02}
and online matching occupies a central position in the field of online
algorithms~\cite{M-TCS12}, there is a long list of work in both these areas. We
state the most relevant works here and refer readers to further related
work in \ifFULL Section~\ref{sec:morerelated}.\else the full version.\fi

Offline matroid intersection was brought to prominence in the groundbreaking work of
Edmonds~\cite{Edmonds-70}. To illustrate the difficulty in moving from
bipartite matching to matroid intersection, we note that while the first linear
time algorithms that beat half for bipartite matching were designed more than 20 years ago~\cite{HK-SICOMP73,ADFS-Journal95}, the fastest known matroid intersection algorithms till today that
beat half make $\Omega(r m)$ calls to the independence
oracles, where $r$ is the rank of the optimal
solution~\cite{ChekuriQuanrud-SODA16,HKK-SODA16}. The quadratic term appears
because  matroid intersection algorithms rely on constructing auxiliary graphs that  needs $\Omega(r m)$ calls~\cite[Chapter 13]{KV08}.  
Until our work,  achieving a competitive ratio
better than half with linear number of independence oracle calls was not known.  The key
ingredient that allows us to circumvent these difficulties is the
\emph{Sampling Lemma} for matroid intersection.
We do not construct an auxiliary graph and instead show that any maximal common
independent is either already a $(\frac12+\delta)$ approximation, or we can
improve it to a $(\frac12+\delta)$ approximation in a single pass over all the
elements.

Online bipartite matching has been studied extensively in the vertex arrival
model (see a nice survey by Mehta~\cite{M-TCS12}). 
Since adversarial arrival order often becomes too pessimistic, the random arrival model
(similar to the \emph{secretary problem}) for online matching was first studied by Goel and Mehta~\cite{GM-SODA08}. 
Since then, this modeling assumption has become standard~\cite{KorulaPal-ICALP09,MY-STOC11,KMT-STOC11,KMZ-STOC15}.
The only progress 
 when edges arrive one-by-one has been in showing lower bounds:
no algorithm can achieve a competitive ratio better than $0.57$ (see
\cite{ELSW-STACS12}), even when the algorithm is allowed to drop edges.

While nothing was previously known for online matching in the random edge
arrival model, similar problems have been studied in the streaming model, most
notably by Konrad et al.~\cite{KMM-APPROX12}. They gave the first algorithm
that beats half for bipartite matching in the
random arrival streaming model. In this work we generalize their 
 \emph{Hastiness Lemma} to matroids. However, prior works on
online matching \ifFULL(see Section~\ref{sec:morerelated}) \else \fi are not  useful as they are tailored to 
graphs---for instance their reliance on notion of ``vertices" cannot
be easily extended to the framework of matroids.  

The simplicity of our OMI algorithm and  flexibility of our analysis allows us
to tackle problems of much greater generality, such as general graphs and
$k$-matroid intersection,  when previously even special cases like bipartite matching had
been considered difficult in the online regime~\cite{MV-PersCom}.  While our
results are a qualitative advance, the quantitative  improvement is small
($\delta > 10^{-4}$).  It remains an interesting challenge to improve the
approximation factor $\delta$. Perhaps a more interesting challenge is to relax
the random order requirement.


\subsection{Our Techniques } \label{sec:idea}

In this section, we present an overview of our techniques to prove 
Theorem~\ref{thm:matroidmain}. Our analysis relies on two observations about
the greedy algorithm that are encompassed in the {Sampling Lemma} and the
{Hastiness Lemma}; 
the latter being useful to extend our linear time offline matroid intersection result to the online setting.
Informally, the Sampling Lemma states that the greedy algorithm cannot  perform
poorly on a randomly generated OMI instance, and the Hastiness Lemma states that 
if the greedy algorithm performs poorly, then it picks most of its elements 
quickly. 

Let  $\OPT$ denote a fixed maximum independent set in the intersection of
matroids $\M_1$ and $\M_2$. WLOG, we assume that the greedy algorithm is
\emph{bad}---returns a common independent set $T$ of size 
$\approx \frac12 |\OPT|$. For offline matroid intersection, by  running the greedy algorithm once,
  one can assume that  $T$ is known. For online matroid intersection, we use the
{Hastiness Lemma} to construct $T$. It states that
 even if we run the greedy algorithm for a  small fraction $f$
(say $<1\%$) of elements, it already picks a set $T$ of elements of size  $ \approx \frac12 |\OPT|$.
This lemma was first observed by Konrad et al.~\cite{KMM-APPROX12} for
bipartite matching and is generalized  to matroid intersection in this work. 
By running the greedy algorithm for this small fraction $f$, the lemma lets
us assume that we start with an approximately maximal common independent set
$T$ with most of the elements  ($1-f > 99\%$) still to arrive.
 
The above discussion reduces the problem  to improving a common independent
set $T$  of size $\approx \frac12 |\OPT|$ to a common  independent set
of size $\geq (\frac12+\delta) |\OPT|$ in a single pass over all the elements.
(This is true for both linear-time offline and OMI problems.) Since $T$ is
approximately maximal, we know that picking most elements in $T$ eliminates the
possibility of picking two $\OPT$ elements (one for each matroid).
Hence, to beat half-competitiveness, we drop a uniformly random $p$
fraction of these ``bad'' elements in $T$ to obtain a set $S$, and try to pick  $(1+\gamma) \OPT$ elements 
(for constant $\gamma>0$) per dropped element. Our main challenge  is to construct an
online algorithm that can  get on average $\gamma$ gain per dropped element
of $T$ in a single pass. The Sampling Lemma for matroid intersection, which is our main 
technical contribution, comes to  rescue. 

\noindent \textbf{Sampling Lemma (informal):} \textit{Suppose $T$ is a common
    independent set in matroids $\M_1$ and $\M_2$, and define $\tE = \cl_1(T)$.
    Let $S$ denote a random set containing each element of $T$ independently
    with probability $(1-p)$.  Then, 
    \[ \E_S[|Greedy(\M_1 /S, \M_2/T,\tE)|] \geq \left(\frac{1}{1+p}\right)
        \cdot \E_S[|\OPT(\M_1 /S, \M_2/T, \tE) |].  \]
}\noindent Intuitively, it says that if we restrict our attention to elements in $\cl_1(T)$ then dropping random elements from $T$ allows us to pick more than $1/(1+p) \geq 1/2$ fraction of the optimal intersection. The advantage over half  yields the $\gamma$ gain per dropped
element. Applying the lemma requires  care as we  apply it twice, once for $(\M_1/S, \M_2/T)$ and once for $(\M_1/T, \M_2/S)$, while ensuring that the resulting solutions have few ``conflicts'' with each other. We overcome this by only considering elements that are in the span of $T$ for exactly one of the matroids.


The proof of the Sampling Lemma involves giving an alternate view of the greedy algorithm
for the random OMI instance. 
Using a carefully constructed invariant and  the method of deferred decisions,
we  show that the expected greedy solution is not too small. 

\IGNORE{\noindent Intuitively, it says that for the random OMI instance on the
contracted matroids $\M_1/S$ and $\M_2/T$, the expected size of greedy is more
than $1/(1+p) \geq 1/2$ fraction of the optimal intersection. For $p<1$, the 
advantage over half  yields the $\gamma$ gain per dropped
element.  The proof of the lemma involves giving an alternate view of the greedy algorithm
for the random OMI instance. 
Using a carefully constructed invariant and  the method of deferred decisions
we  show that the expected greedy solution is not too small. 

Applying the Sampling Lemma requires  care as we need to apply it twice, once for $(\M_1/S, \M_2/T)$ and once for $(\M_1/T, \M_2/S)$, while ensuring that the resulting solutions have few ``conflicts'' with each other. To overcome this, we only consider elements that are in the span of $T$ for exactly one of the matroids.  }

\ifFULL
\subsection{Further Related Work}\label{sec:morerelated}
\noindent \textbf{Online Matching in Vertex Arrival Model}\\
Karp, Vazirani, and Vazirani~\cite{KVV-STOC90}  presented the ranking algorithm
for online bipartite matching in the  vertex arrival model. The problem is to
find a matching
in a bipartite graph where one side of the bipartition is
fixed, while the other side vertices arrive in an online fashion.
Upon arrival
of a vertex, its edges to the fixed vertices are revealed, and the algorithm
must immediately and irrevocably decide where to match it.  \cite{KVV-STOC90}
gives an optimal $\left(1-\frac1e\right)$-competitive ranking algorithm 
for adversarial vertex arrival. Since their original proof was incorrect, new ways of analyzing the ranking algorithm have since been
developed~\cite{BM08-ASN,DJK-SODA13}. 
Due to its many applications in the online ad-market,
 the vertex arrival model, its weighted generalizations, and vertex arrival on both sides have been studied thoroughly (see survey~\cite{M-TCS12,WangWong-ICALP15}). 

Goel and Mehta~\cite{GM-SODA08} introduced the random vertex-arrival model.
In this model, the adversary may choose the worst instance of a graph, but the
online vertices  arrive in a random order. The greedy algorithm is already
$\left(1-\frac1e\right)$-competitive for this problem, as the analysis reduces
to~\cite{KVV-STOC90}.  Later works~\cite{MY-STOC11,KMT-STOC11} 
showed that the ranking algorithm has a competitive ratio of at least $0.69$,
beating the bounds for adversarial vertex arrival model. There is still a gap
between known upper and lower bounds, and closing this gap remains an open
problem. 

\noindent \textbf{Online Matching in Edge Arrival Model}\\
In the edge arrival model, a fixed
bipartite  graph is  chosen by an adversary and its edges are revealed one by
one to an online algorithm that is trying to find a maximum matching.
If the edge arrival is adversarial, this problem
captures the adversarial vertex arrival model as a special case: 
constraint the edges incident to a vertex to appear together.  The 
greedy algorithm has a competitive ratio of half and a natural open question
is whether we can beat half. The current best hardness result for
adversarial edge arrival is $\sim 0.57$, even when the
algorithm is allowed to drop edges (see~\cite{ELMS-SJDM11}). 

Matching in the edge arrival model has also been studied in the streaming
community. In the streaming model, the matching algorithm can revoke decisions
made earlier, but has only a bounded memory; in particular, it has
$\thicktilde{O}(1)$ memory in the streaming model and $\thicktilde{O}(n)$ memory in the
semi-streaming model (see~\cite{FKMSZ-TCS05}).  The algorithm may
make multiple passes over the input; usually trading off the number of passes
with the quality of the solution.  
For bipartite matching in adversarial edge arrival, 
Kapralov~\cite{Kapralov-SODA13-Streaming} showed that no semi-streaming matching algorithm can do better than $1 - \frac1e$. Beating half remains a major open problem. 

On the other hand, for uniformly random edge arrival 
Konrad, Magniez, and Mathieu~\cite{KMM-APPROX12} gave the first single pass algorithm that
obtains a $0.501$-competitive ratio for bipartite matching in the semi-streaming setting.  Their algorithm  crucially used the ability to revoke earlier
decisions. One of the contributions in this paper is to show that a variant of the 
greedy algorithm, which appears simple in hindsight, achieves a competitive
ratio better than half in the more restrictive online model.

A weighted generalization of OBME is online bipartite matching for random edge
arrival in an edge weighted bipartite graph. This problem has exactly the same
setting as OBME; however, the goal is to maximize the weight of the matching
obtained.  Since it is a generalization of the secretary problem, the greedy
algorithm is no longer constant competitive.  Korula and
Pal~\cite{KorulaPal-ICALP09} achieved a breakthrough and gave a constant
competitive ratio algorithm for this problem\footnote{They also obtain similar results for hypergraphs and call it the
    ``Hypergraph Edge-at-a-time Matching'' problem.}. Kesselheim et al.~\cite{KRTV-ESA13} later improved their results.

 \noindent \textbf{Randomized Greedy Matching Algorithms}\\
Our result for matching in general graphs follows a line of work analyzing variants of the greedy algorithm for matching in general graphs. Dyer and
Frieze~\cite{DF-RANDOM91} showed that greedy on a uniformly random permutation of the edges 
cannot achieve a competitive
ratio better than half for general graphs; however, it performs well for some
classes of sparse graphs. Aaronson et al.~\cite{ADFS-Journal95} proposed the
Modified Randomized Greedy (MRG) algorithm and showed that it has a competitive
 ratio better than half for general graphs. 
Poloczek and Szegedy~\cite{PS-FOCS12} provided an argument to improve the bounds on the competitive ratio of this algorithm; however, a gap has emerged in their contrast lemma. 
A ranking based randomized greedy algorithm has been also shown to have a competitive ratio
better than half for general graphs (see~\cite{CCWZ-SODA14}).
Neither MRG nor the ranking algorithm can be implemented in the
original setting of \cite{DF-RANDOM91} where the edges arrive in random
order and the algorithm is only allowed a single pass.  To prove Theorem~\ref{thm:generalgraphs}, we give an algorithm that beats greedy for  general graphs with a much simpler analysis and also 
works in the original setting of \cite{DF-RANDOM91}.

\noindent \textbf{Online Matroid Problems}\\
The OMI problem studied in this paper is much more general than online matching and has many other applications, such as the following online network design problem. Consider a central depot that
stores different types of commodities and is connected to different cities by
rail-links. At various points cities order one of the commodities from the
depot and the central manager must immediately and irrevocably decide  whether
to fulfill the order.  If the central manager chooses to fulfill the order, it
needs to find a path of rail-links from the depot to that city. Moreover, any
rail-link can be used to fulfill at most one order as it can only run a single
train.  The question is to maximize the number of accepted requests given that
there is only a finite amount of each commodity at the depot.  This is a
matroid intersection problem between a gammoid and a partition matroid.  Our
result implies an algorithm that beats half for this
problem if the orders arrive uniformly at random. The
intersection of two graphic matroids, with applications to electrical
networks~\cite{Recski-DM05}, is another special case of matroid
intersection that has received attention in the past~\cite{GabowXu-Journal96}.

The uniformly random order assumption in OMI is motivated from the work on the secretary problem. 
In 2007, Babaioff, Immorlica, and Kleinberg~\cite{BIK-SODA07} introduced the matroid
secretary problem, which generalized the classical secretary problem. For a
matroid with weighted elements arriving in a uniformly random order, the online
algorithm needs to select an  independent set of large weight. 
Despite recent breakthroughs (see~\cite{Lachish-FOCS14,FSZ-SODA15}), their question
on the existence of a constant-competitive algorithm remains unanswered. This problem becomes 
trivial in the unweighted setting as the greedy algorithm finds the optimum solution. However,  beating greedy
remained challenging for intersection of matroids. Our Theorem~\ref{thm:matroidmain} resolves
this problem.  For weighted online matroid intersection, constant factor competitive
algorithms are known in the streaming model where the algorithm always maintains an
independent set in the intersection but is allowed to drop elements
(see~\cite{Varadaraja-ICALP11}).


\noindent \textbf{Offline Matroid Intersection}\\
Until recently, the fastest unweighted offline matroids intersection algorithm was
a variant of Hopcraft-Karp bipartite matching algorithm due to
Cunningham~\cite{Cunningham86} taking  $O(m r^{3/2} Q)$ time --- $m,r$, and $Q$ refer to the
number of ground elements, the rank of matroid intersection, and to the
independence oracle query time, respectively. In 2015, Lee, Sidford, and
Wong~\cite{LSS-FOCS15} improved this to $\thicktilde{O}(m^2 Q + m^3)$, both for
weighted and unweighted matroid intersection.  Not much success has been achieved in  proving  lower bounds on the oracle complexity of matroid intersection algorithms~\cite{Harvey-SODA08}.
When looking for a $(1-\epsilon)$~approximate
weighted matroid intersection, recent works have improved the running time
to $\thicktilde{O}(mrQ/{\epsilon^2} )$~\cite{ChekuriQuanrud-SODA16,HKK-SODA16}. Our main Theorem~\ref{thm:matroidmain}  gives the first algorithm that achieves an approximation factor greater than half with only a linear number of calls to the independence oracles,  i.e., in $O(m\,Q)$ time.
 
 \else
 \fi


\section{Warmup: Online Bipartite Matching} \label{section:matching}
In this section, we consider a special case of online matroid intersection,
namely online bipartite matching in the random edge arrival model.  Although,
this is a special case of the general Theorem~\ref{thm:matroidmain}, we present
it because  nothing non-trivial was known before \ifFULL(see
Section~\ref{sec:morerelated}) \else \fi and several of our ideas greatly simplify in
this case (in particular the Sampling Lemma), allowing us to lay the framework
of  our ideas.


\subsection{Definitions and Notation} \label{sec:notmatching}
An instance of the online bipartite matching problem $(G,E,\pi,m)$ consists of a 
bipartite graph $G=(U\cup V, E)$ with $m=|E|$, and where the edges in $E$
arrive according to the order defined by $\pi.$ We assume that the algorithm
knows $m$  but does not know $E$ or $\pi$.  For $1\leq i\leq j \leq m$, let 
$E^{\pi}[i,j]$   denote
the set of edges that arrive in between positions $i$ through $j$ according to
$\pi$ \footnote{We emphasize that our definition also works when $i$ and $j$
    are non-integral}. When  permutation $\pi$ is implicit, we 
abbreviate this to $E[i,j]$. 

\Greedy denotes the algorithm that picks an edge into the matching whenever possible. 
Let $\OPT$ denote a fixed maximum offline matching of graph $G$.  For $f
\in [0,1]$, let $T^{\pi}_f$ denote the matching produced by \Greedy after seeing the first $f$-fraction of the edges according to order
$\pi$. For a uniformly random chosen order $\pi$,  
    \[ \G(f):= \frac{ \E_{\pi}[ |T^{\pi}_f| ] }{|\OPT|}.\] 
Hence, $\G(1)\,|\OPT|$ is the expected output size of \Greedy and
$\G(\frac12)\,|\OPT|$ is the expected output size of \Greedy after seeing half of the
edges.  We observe that  \Greedy has a competitive ratio of at-least half and in \ifFULL
Section~\ref{appn:otherResults}, \else the full version \fi we show that this ratio is tight for worst case
input graphs.\footnote{We also show that for regular graphs \Greedy is at least
$\left( 1- \frac1e \right)$  competitive, and that no online algorithm for OBME can be better than $\frac{69}{84} \approx 0.821$ competitive. }



\subsection{ Beating Half }

\ifFULL

\begin{figure}
	\centering
	\begin{tikzpicture} [thin,scale=1.5]
		\draw (0,-0.05) ellipse (0.2cm and 0.60cm)[fill=black!10];	
		\draw (1.5, -0.1) ellipse (0.2cm and 0.5cm)[fill=black!10];
		
		\draw [-, densely dashdotted, thick] (0,0.2) to node[auto]{$ $} (1.5,0.2);
		\draw [-, densely dashdotted, thick] (0,-0) to node[auto]{$ $} (1.5,0);

		\draw (1.5,-1.55) ellipse (0.2cm and 0.60cm)[fill=black!10];
		\draw (0,-1.5) ellipse (0.2cm and 0.5cm)[fill=black!10];

		\draw (1.5,0.1) ellipse (0.1cm and 0.21cm)[fill=blue!40];
		\draw (0,-1.3) ellipse (0.1cm and 0.21cm)[fill=blue!40];
		
		\draw [-,  dashed, thick, red] (0,-1.2) to node[auto]{$ $} (1.5,0.2);
		\draw [-,  dashed, thick, red] (0,-1.4) to node[auto]{$ $} (1.5,-0);
		
		\draw [-, densely dashdotted, thick] (0,-1.2) to node[auto]{$ $} (1.5, -1.2);
		\draw [-, densely dashdotted, thick] (0,-1.4) to node[auto]{$ $} (1.5, -1.4);
		
		\draw [-,   thick] (0,-1.65) to node[auto]{$ $} (1.5,-0.25);
		\draw [-,   thick] (0,-1.85) to node[auto]{$ $} (1.5,-0.45);		

		\draw (-0.7,0.6) -- (2.2,0.6) -- (2.2, -0.7) -- (-0.7, -0.7) -- cycle;		
		\draw (-0.7,-0.9) -- (2.2,-0.9) -- (2.2, -2.2) -- (-0.7, -2.2) -- cycle;	
		
		\node at (-0.5,-0) {$Y_2$};
		\node at (2,-0) {$X_2$};
		\node at (2.5,0) {$G_2$};
		\node at (2.5,-1.5) {$G_1$};
		\node at (-0.5,-1.5) {$X_1$};
		\node at (2,-1.5) {$Y_1$};

	\end{tikzpicture}
    \caption[margin=5cm]{$U= X_1 \cup Y_2$ and $V=X_2\cup Y_1$, where $X_1$ and $X_2$ denote the set of vertices matched by \Greedy in Phase~(a). Here thick-edges are picked and diagonal-dashed-edges
    are  marked. Horizontal-dashed-edges show augmentations for the
    marked edges. } 
        \label{figOBME}
\end{figure}
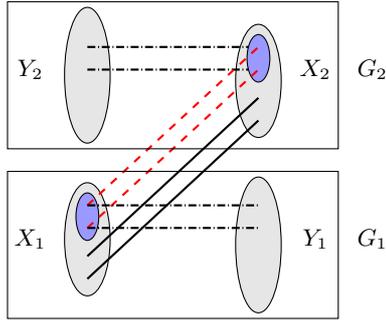		

\else

\fi

Lemma~\ref{lem:reduc} shows that we can restrict our attention to the case when
the expected \Greedy size is small\ifFULL \else(proof in the full version)\fi. 
Theorem~\ref{thm:onebit}  gives an algorithm that beats half for this
restricted case.

\begin{lemma}\label{lem:reduc}
Suppose there exists an Algorithm~\A that achieves a competitive ratio of  $\frac12 + \gamma$
when $\G(1) \leq (\frac{1}{2}+\epsilon)$ for some
$\epsilon, \gamma >0$. Then there exists an algorithm with competitive ratio at least 
$\frac12 + \delta $, where $\delta = \frac{\epsilon  \gamma}{\frac12 + \epsilon + \gamma}$.
\end{lemma}
\ifFULL
\begin{proof}
    Consider the algorithm that tosses a coin at the beginning and runs \Greedy
    with probability $1-r$ and Algorithm~$\A$ with probability
    $r$, where $r>0$ is some constant. This lemma follows from simple case analysis.
    \begin{itemize}
        \item Case 1: $\G(1) < \frac12 + \epsilon$ \newline
            Since \Greedy is always $\frac12$ competitive, we can say that in expectation, 
            the competitive ratio will be at least
                \[ \left(1-r \right) \frac12 + r \left(\frac12 + \gamma \right) 
                        = \frac12 + r \gamma \] 
        \item Case 2: $\G(1) \geq \frac12 + \epsilon$ \newline
            Since we have no guarantees on the performance of Algorithm~$\A$
            when \Greedy performs well, we assume that it achieves a competitive
            ratio of 0. Our expected performance will be at least
                \[ \left(1 - r \right) \left(\frac12 + \epsilon \right)
                    + 0 = \frac12 +  \epsilon - \frac{r}{2} - r \epsilon \]
    \end{itemize}
    Choosing $r = \frac{\epsilon}{\frac12 + \epsilon + \gamma}$, we get 
    $\delta \geq \frac{\epsilon  \gamma}{\frac12 + \epsilon + \gamma}.$
\end{proof}
\else
\fi

\begin{theorem} \label{thm:onebit}
    If $\G(1) \leq (\frac{1}{2}+\epsilon)$ for some constant $\epsilon>0$ then the
    \MkGreedy algorithm outputs a matching of size at least
    $(\frac{1}{2}+\gamma)\,|\OPT|$ in  expectation, where $\gamma>0$ is a constant.
\end{theorem}

Before describing \MkGreedy, we need the following property about the
performance of \Greedy in the random arrival model --- if \Greedy is bad then it
makes most of its decisions quickly and incorrectly.  
We will be interested in the regime where $0 < \epsilon \ll f \ll 1/2$.

\begin{lemma}[Hastiness property: Lemma~$2$ in~\cite{KMM-APPROX12}] \label{lemma:KMMmatch}
    For any graph $G$ if $\G(1) \leq (\frac{1}{2}+ \epsilon)$ for some
    $0<\epsilon<\frac{1}{2}$, then for any $0<f<1/2$ 
    \[ \G(f) \geq \frac{1}{2} - \left( \frac{1}{f} -2 \right)\epsilon. \]
\end{lemma}

\subsubsection{\MkGreedy  for bipartite matching: } \label{sec:algomatch}
~

\noindent \MkGreedy consists of two phases  (see the pseudocode). 
In Phase~(a), it runs \Greedy for the 
first $f$-fraction of the edges, but \emph{picks} each edge \emph{selected} by
\Greedy into the final matching only with probability $(1-p)$, where $p>0$ is a
constant. With the remaining probability $p$, it \emph{marks} the edge $e$ and
its vertices, and behaves as if it had been picked.  In Phase~(b), which is for
the remaining $1-f$ fraction of edges, the algorithm runs \Greedy to pick edges
on two restricted disjoint subgraphs $G_1$ and $G_2$, where  it only considers
edges incident to exactly one marked vertex in Phase~(a). \ifFULL(see
Figure~\ref{figOBME}.)\else \fi 

Phase~(a) is equivalent to running \Greedy to select elements, but then
randomly dropping  $p$ fraction of the selected edges.  The idea of marking
some vertices (by marking an incident edge) is to ``protect'' them for
augmentation in Phase~(b). To distinguish if an edge is marked or picked, the
algorithm uses auxiliary random bits $\K$ that are unknown to the adversary. We
assume that $\K(e) \sim \text{Bern}(1-p)$ i.i.d. for all $e \in E$.
\vspace{-0.5cm}
\begin{algorithm}
\caption{$\MkGreedy(G,E,\pi,m,\K)$}
\label{AlgA1}
\begin{algorithmic}[1]
    \Statex \textbf{Phase~(a)}
    \State Initialize $S, T, N_1, N_2 $ to $\emptyset$
    \For{each element $e \in E^{\pi}[1,fm]$} \Comment{\emph{\Greedy while picking and marking}}
       \If {$T \cup e $ is a matching in $G$}
            \State $T \gets T \cup e$ \Comment{\emph{Elements selected by \Greedy}}
            \If{$\K(e)=1$ }	\Comment{\emph{Auxiliary random bits $\K$}}
                \State $S  \gets S \cup e $ \Comment{\emph{Elements picked into final solution}}
            \EndIf
       \EndIf
    \EndFor
    \Statex \textbf{Phase~(b)}
    \State Initialize set $T_f$ to $T$. Let sets $X_1, X_2$ be vertices of  $U, V$ matched in $T_f$ respectively.
    \State Let $G_1$ be the subgraph of $G$ induced on $X_1$ and $V \setminus X_2$. 
    \State Let $G_2$ be the subgraph of $G$ induced on $U \setminus X_1$ and $X_2$.
    \For{ each edge $e \in (E^{\pi}[fm,m])$} \Comment{\emph{\Greedy on two disjoint subgraphs}}
    	\For{ $i \in \{1,2\}$}
            \If {$e \in G_i$ and $S\cup N_i \cup e $ is a matching }  \Comment{\Greedy step} 
                \State $N_i \gets N_i \cup e$  \Comment{\emph{New edges picked}}
            \EndIf
        \EndFor
    \EndFor
    \State \Return $S \cup N_1 \cup N_2$
\end{algorithmic}
\end{algorithm}

\paragraph{Comparison to Konrad et al.~\cite{KMM-APPROX12}}
For the special case of bipartite matching, we can consider \MkGreedy  to be a
variant of the streaming algorithm of \cite{KMM-APPROX12}.  For graphs where
\Greedy is bad,  both algorithms use the first phase to pick an approximately
maximal matching $T$ using {the Hastiness Lemma}.  \cite{KMM-APPROX12} divides
the remaining stream into two portions and uses each portion to find greedy
matchings, say $F_1$ and $F_2$. Since decisions in the streaming setting are
revocable, at the end of the stream they use  edges in $F_1 \cup F_2$ to  find
sufficient number of three-augmenting paths w.r.t.  $T$.  Their algorithm is
not online because it keeps all the matchings till the end.  
One can view the current algorithm as turning their algorithm  into
an online one by flipping a coin for each edge in $T$.  In the second phase,
it runs \Greedy on  two random disjoint
subgraphs and use {the Sampling Lemma} to argue that in expectation the
algorithm picks sufficient number of augmenting paths.  

While our online matching algorithm is simple and succinct, the main difficulty
lies in extending it to OMI as the notions of marking and protecting vertices
do not exist.  This is also the reason why obtaining a linear time algorithm
for offline matroid intersection problem, where  Hastiness Lemma is not needed,
had been open.  Defining and proving the correct form of Sampling Lemma forms
the core of our OMI analysis  in Section~\ref{section:matroids}.


\subsubsection{Proof that \MkGreedy works for bipartite matching:} 
\label{sec:analysisMatch}~

\noindent Let $G_i$ denote graphs $G_1$ or $G_2$ for $i\in\{1,2\}$. 
For a fixed order $\pi$ of the edges, graphs $G_i$ in \MkGreedy are independent
of the randomness $\K$. Since the algorithm uses $\K$ to pick a random subset
of the \Greedy solution, this can be viewed as independently sampling each
vertex matched by \Greedy in $G_i$. Lemma~\ref{lemma:samp} shows that this
suffices to pick in expectation more than the number of marked edges.  In
essence, we use the randomness $\K$ to limit the power of an  adversary
deciding the order of the edges in Phase~(b). While the proof follows from 
the more general Lemma~\ref{lem:matrmain}, we include a simple self-contained proof\ifFULL\else  for this case in the full version\fi. 


\begin{lemma}[Sampling Lemma\footnote{This special case of Lemma~\ref{lem:matrmain} for  bipartite matching  is also present in~\cite{KMM-APPROX12}. The authors of this paper thank Deeparnab Chakrabarty for pointing this  at IPCO 2017. }] \label{lemma:samp} 
    Consider a bipartite graph $H = ( X \cup Y, \tE)$  containing  a matching $\tI$.
    Let $\K(x) \sim \text{Bern}(1-p)$ i.i.d. for all $x\in X$, and define $X' = \{ x \mid x \in X \text{ and } \K(x)=0 \}.$ 
    I.e., the vertices of $X'$ are obtained by independently sampling each vertex in $X$ with probability $p$. 
    Let $H'$ denote the subgraph induced on $X' $ and $Y$. 
    Then  for any arrival order of the edges in $H'$,  
    \[ \E_{\K}[ \Greedy(H', \tE) ] \geq \frac{1}{1+p} \left( p |\tI| \right) . \]
\end{lemma}
\ifFULL
\begin{proof} 
   We prove this statement by induction on $|\tI|$. 
   Consider the base case $|\tI|=1$. Whenever
   \Greedy does not select any edge, the vertex adjacent to $\tI$ in $X$ is not sampled. 
   This happens with probability $1-p$. Hence, the expected size of the matching is
   at least $p \geq \frac{p}{1+p}$, which implies the statement is true when $|\tI|=1.$

   From the induction hypothesis (I.H.) we can assume  the statement is true  
   when the matching size is at most $|\tI|-1$.
   We prove the induction step by contradiction and consider the smallest graph in terms of $|X|$ 
   that does not satisfy the statement.  Note that $|X| \geq |\tI|$.
   Consider the first edge $e=(x,y)$ that arrives. 
   The first case is when $x \not\in X'$ and it happens with probability $1-p$. Here any edge incident
   to $x$ does not matter for the remaining algorithm.   We use I.H. 
   on the subgraph induced on $(X\backslash x, Y)$ as
   $|X \backslash x| = (|X|-1)$. Since this subgraph has a matching of size at least $|\tI| -1$, I.H. gives a  matching of expected size at least
    $\frac{p}{1+p} ( |\tI|-1)$.

    The second case is when $x \in X'$  and it happens with probability $p$. Now edge $(x,y)$ is
    included in the \Greedy matching for the induced graph on $(X',Y)$.  Vertices
    $x$ and $y$, along with the edges incident to them, do not participate in the
    remaining algorithm. We apply I.H. on the subgraph
    induced on the vertices $(X\backslash x , Y\backslash y)$. Noting that
    this graph has a matching of size at least $|\tI|-2$, I.H. gives a matching of expected
    size at least $\frac{p}{1+p} (|\tI|-2)$. 
    Combining both cases, the expected matching size is at least \[(1-p)
        \left(\frac{p}{1+p} (|\tI|-1) \right) + p \left( 1 + \frac{p}{1+p} (|\tI|-2)
        \right) = \frac{p}{1+p} |\tI|. \]
    This is a contradiction as we assumed that the graph did not satisfy the induction statement, which completes the proof of Lemma~\ref{lemma:samp}.
\end{proof}
\else
\fi

We next prove the main lemma needed to prove Theorem~\ref{thm:onebit}. Setting $f=0.07$,
$p=0.36$, and $\epsilon = 0.001$ in Lemma~\ref{lemma:main}, the theorem
 follows by taking $\gamma > 0.05$. 
\begin{lemma} \label{lemma:main} 
    For any $0 < f< 1/2$  and bipartite graph $G$, \MkGreedy outputs a matching
    of expected size at least 
    \[ \left[ (1-p) \left( \frac12 - \left( \frac1f - 2 \right)\epsilon \right) +  \frac{p}{1+p}   \left( 1 - \frac{2\epsilon}{f} - f\right)  \right] |\OPT|. \]
\end{lemma}
\begin{proof}
    We remind the reader that for any $f \in [0,1]$ and any permutation $\pi$ of the edges,  $T_f^{\pi}$     denotes the matching that \Greedy produces on $E^{\pi}[1,fm]$.    
    For $i \in \{1,2\}$, let $H_i$ denote the subgraph of $G_i$ containing all its edges that appear in Phase~(b). 
Let $\I_i$ denote the set of edges of $\OPT$ that appear in graph $G_i$. We use the following  claim\ifFULL\else ~proved in the full version\fi.
\begin{myclaim}\label{claim:bipmatchsizeI}
\[ \E_\pi \left[ |\I_1| + |\I_2| \right] \geq \left(1 - \frac{2\epsilon}{f} \right)\,|\OPT|. \]
\end{myclaim}
\ifFULL
\begin{proof}
We use the following two simple properties of $T_1^{\pi}$ (proved in Section~\ref{sec:missfacts}).
\begin{fact}\label{fact:matching}
    \label{fact:matching}
    \begin{align}
        |T_1^{\pi}| &\geq  \frac12 \left( |\OPT| + \sum_{e \in \OPT}  
                \mathbf{1}[\text{Both ends of $e$ matched in $T_f^{\pi}$}] \right) \text{ and }
                                    \label{eq:factT11}\\
        |T_1^{\pi}| &\geq |T_f^{\pi}| + \frac12 \sum_{e \in \OPT}  
            \mathbf{1}[\text{Both ends of $e$ unmatched in $T_f^{\pi}$}].  \label{eq:factT12} 
    \end{align}
\end{fact}

\noindent Note that, $\E_{\pi} \left[ |\I_1| + |\I_2| \right]$ is equal to
\begin{align*}
& \hspace{-0.7cm} \E_{\pi} \left[ |\OPT| - \sum_{e \in \OPT}  \mathbf{1}[\text{Both ends of $e$ matched in $T_f^{\pi}$}]  \right. \\
& \hspace{2.5cm} \left. - \sum_{e \in \OPT}  \mathbf{1}[\text{Both ends of $e$ unmatched in $T_f^{\pi}$}]  \right] \\
\geq~ &  |\OPT| - \E_{\pi} \left[ 2\,|T_1^{\pi}| - |\OPT| \right] -  \E_{\pi} \left[ 2(|T_1 ^{\pi}| - |T_f^{\pi}|)\right]  \tag{using Eq.~(\ref{eq:factT11}) and Eq.~(\ref{eq:factT12})} \\
\geq~ & |\OPT| - 2\epsilon\,|\OPT| - 2\left( \epsilon +  \left(\frac1f -2 \right) \epsilon \right)\,|\OPT| \tag{using $\G(1) \leq \frac12+\epsilon$ and Lemma~\ref{lemma:KMMmatch}} \\
=~ & \left(1 - \frac{2\epsilon}{f} \right)\,|\OPT|, \text{ which finishes the proof of the claim.}	 
\end{align*}
\end{proof}
\else
\fi

For $i \in \{1,2\}$, let $\tI_i \subseteq \I_i$ denote the set of edges of
$\OPT$ that appear in Phase~(b) of \MkGreedy, i.e., they appear in graph $H_i$.
In expectation over uniform permutation $\pi$, at most $f|\OPT|$ elements of
$\OPT$ can appear in Phase~(a). Hence, 
\[  \E_{\pi} \left[ |\tI_1| + |\tI_2| \right] \geq \E_{\pi} \left[ |\I_1| + |\I_2| \right] - f|\OPT| 
    \geq \left( 1 - \frac{2\epsilon}{f} - f\right)\,|\OPT|.\]
Marking a random subset of
    $T_f^{\pi}$ independently is equivalent to marking a random subset of vertices
    independently. Thus, we can apply Lemma~\ref{lemma:samp} to both $H_1$ and $H_2$.
    The expected number of edges in $N_1 \cup N_2$ is at least  $
    \frac{p}{1+p}(|\tI_1| + |\tI_2| )$, where the expectation is over the auxilary bits $\K$ that distinguishes
    the random set of edges marked. Taking expectations over $\pi$
    and noting that Phase~(a) picks $(1-p)\,\G(f)\,|\OPT|$ edges, we have 
 \begin{align*} 
&  \E_{\K,
            \pi}[|S\cup N_1 \cup N_2|] \quad = \quad \E_{\K, \pi}[|S|] +  \E_{\K, \pi}[|N_1| + |N_2|] \\
& \geq ~ \G(f)  (1-p)\,|\OPT| + \frac{p}{1+p} \E_{\pi} \left[|\tI_1| + |\tI_2| \right] \\
 &\geq ~ \left[ (1-p) \left( \frac12 - \left( \frac1f - 2 \right)\epsilon \right) +  \frac{p}{1+p}   \left( 1 - \frac{2\epsilon}{f} - f\right)  \right] |\OPT| \quad \left(\text{by Lemma~\ref{lemma:KMMmatch}} \right).  
\end{align*}
\end{proof}


\section{Online Matroid Intersection}\label{section:matroids}

\subsection{Definitions and Notation} \label{sec:matrnotation}
  An instance of the online matroid intersection problem 
$(\M_1, \M_2, E, {\pi},m)$ consists of matroids $\M_1$ and $\M_2$ defined on 
ground set $E$ of size $m$, and where the elements in  $E$ arrive according to the order defined by $\pi$.
For any $1 \leq i \leq j \leq m$, let $E^{\pi}[i,j]$  denote the
ordered set of elements of $E$ that arrive in positions $i$ through $j$ according to $\pi$.
For any matroid $\M$ on ground set $E$, we use $T \in \M$ to denote $T \subseteq E$ is an independent set in matroid $\M$.
We use the terminology of matroid restriction and matroid contraction as defined in Oxley~\cite{O-MT06}.
To avoid clutter, for any $e \in E$ we  abbreviate $A \cup \{e \}$ to $A \cup e$ and $A \setminus \{e\}$ to $A \setminus e$.  

\vspace{-0.5cm}
\begin{algorithm}
\caption{\Greedy$(\M_1,\M_2,E, \pi)$}
\label{AlgGreedy}
\begin{algorithmic}[1]
    \State Initialize set $T$ to $\emptyset$
    \For{each element $e \in E^{\pi}[1,|E|]$}
 	   \If{$T \cup e \in \M_1 \cap \M_2$ }
           \State $T  \gets T \cup e $
        \EndIf
    \EndFor\\
    \Return $T$
\end{algorithmic}
\end{algorithm}
\vspace{-0.5cm}

We note that \Greedy is well defined even when matroids $\M_1$ and $\M_2$ are defined on larger ground sets as long as they contain $E$. This notation will be useful when we run \Greedy on matroids after contracting different sets in the two matroids.
Since \Greedy always produces a maximal independent set, its  competitive ratio is at least half~(see Theorem~13.8 in~\cite{KV08}). This is  because  an ``incorrect'' element  creates at most two circuits in $\OPT$, one for each matroid.

 Let $\OPT$ denote a fixed maximum offline independent set in the intersection
of both the matroids.  For $f \in [0,1]$, let $T^{\pi}_f$
denote the independent set that \Greedy produces after seeing the
first $f$ fraction of the edges according to order $\pi$. When clear from context,
we will often abbreviate $T^{\pi}_f$ with $T_f$.  
Let   $\G(f):= \frac{ \E_{\pi}[ |T_f| ] }{|\OPT|}$, where $\pi$ is a uniformly random chosen order.

For $i \in \{1,2\}$, let $\cl_i(T) := \{ e \mid (e\in E) \wedge
\left(\rank_{\Mi}(T\cup e) = \rank_{\Mi}(T) \right) \}$  denote the span of set
$T \subseteq E$ in matroid $\Mi$. Suppose we have  $T \in \Mi$ and $e \in \cl_i(T)$, then we denote the
unique circuit of $T \cup e$ in matroid $\Mi$ by $C_i(T\cup e)$. If
$i=1$, we use $\thickbar{\imath}$ to denote $2$, and vice versa.  

\ifFULL We provide a table of all notation used in Section~\ref{sec:notation}.
\else \fi

\subsection{Hastiness Property}

Before describing our algorithm \MkGreedy, we need an important hastiness property of \Greedy in the random arrival model. Intuitively, it states that if \Greedy's performance is bad then it
makes most of its decisions quickly and incorrectly.  This observation was
first made by Konrad et al.~\cite{KMM-APPROX12} in the special case of
bipartite matching. We extend this property to
matroids in Lemma~\ref{lemma:KMM} (proof in \ifFULL Section~\ref{sec:lemhastpf}\else the full version\fi). We
are interested in the regime where $0 < \epsilon \ll f \ll 1$.

\begin{lemma}[Hastiness Lemma] \label{lemma:KMM} 
  For any two matroids $\M_1$ and $\M_2$ on the same ground set $E$,  let
  $T_f^{\pi}$ denote the set selected by \Greedy after running for the first
  $f$ fraction of elements $E$ appearing in order $\pi$. Also, for $i\in
  \{1,2\}$, let $\Phi_i(T_f^{\pi}) := \cl_i(T_f^{\pi}) \cap \OPT$.     Now for
  any $0 < f, \epsilon \leq \frac12$, if $\E_{\pi}[|T_1^{\pi}|]  \leq
  (\frac{1}{2}+ \epsilon)\,|\OPT|$ then 
    \begin{align*}
     & \E_{\pi} \left[ |\Phi_1(T_f^{\pi}) \cap \Phi_2(T_f^{\pi})| \right]  
            \leq 2\epsilon\,|\OPT| \qquad \text{and} \\
     & \E_{\pi} \left[ |\Phi_1(T_f^{\pi}) \cup \Phi_2(T_f^{\pi})| \right]  
            \geq \left(1 -  \frac{2\epsilon}{f} + 2 \epsilon \right)\,|\OPT|.   
      \end{align*}
    This implies  $\G(f) := \frac{\E_{\pi}[|T_f^{\pi}|] }{|\OPT|} \geq \left(\frac{1}{2} - 
                \left( \frac{1}{f} -2 \right)\epsilon \right)$.       
 \end{lemma}
                

\subsection{Beating Half for Online Matroid Intersection}

Once again, we use Lemma~\ref{lem:reduc} to restrict our
attention to the case when the expected size of \Greedy is small. 
In Theorem~\ref{thm:onebitmat},  we give an
algorithm that beats half for this restricted case, which when combined with
Lemma~\ref{lem:reduc} finishes the proof of Theorem~\ref{thm:matroidmain}.
%

\begin{theorem} \label{thm:onebitmat}
    For any two matroids $\M_1$ and $\M_2$ on the same ground set $E$, there 
 exist  constants $\epsilon, \gamma>0$ and a randomized online algorithm \MkGreedy  such that if $\G(1) \leq
    \left(\frac{1}{2}+\epsilon \right)$ then $\MkGreedy$ outputs an independent
    set in the intersection of both the matroids of expected size at least $\left(
    \frac{1}{2}+\gamma \right)|\OPT|$.  
\end{theorem}


\subsubsection{\MkGreedy for OMI:} \label{sec:matrOutline}~
\vspace{-0.5cm}
\begin{algorithm}
\caption{\MkGreedy$(\M_1,\M_2,E, \pi, m, \K)$}
\label{AlgA}
\begin{algorithmic}[1]
   \Statex \textbf{Phase~(a)}
    \State Initialize $S,T$ to $\emptyset$
    \For{each element $e \in E^{\pi}[1,fm]$}		\Comment{\emph{\Greedy while picking and marking}}
        \If{ $T \cup e \in \M_1 \cap \M_2 $}
        \State $T \gets T \cup e$	\Comment{\emph{Elements selected by \Greedy}}
            \If{ $\psi(e)=1$ }	\Comment{\emph{Auxiliary random bits $\K$}}
                \State $S \gets  S \cup  e$	\Comment{\emph{Elements picked into the final solution}}
            \EndIf
        \EndIf
    \EndFor
    \Statex \textbf{Phase~(b)}
    \State Fix ${T_f}$ to $T$ and initialize sets ${N}_1, {N}_2$ to $\emptyset$
    \For{each element $e \in E^{\pi}[fm,m]$}	\Comment{\emph{\Greedy on two disjoint problems}}
    	\For { $i \in \{1,2\} $ }
            \If { $e \in \cl_i(T_f )$ and $ e \notin \cl_{\thickbar{\imath}}(T_f)$ } \label{algoB:span} \Comment{\emph{To ensure disjointness}}
                \If { $(S\cup \Ni \cup  e  \in \Mi )$ and 
                            $( T_f\cup \Ni \cup e \in \Mni)$} \Comment{\Greedy step} \label{SamAlg:claimprecon}
                    \State $\Ni \gets \Ni \cup e$	\Comment{\emph{Newly picked elements}}
                \EndIf
            \EndIf
        \EndFor
    \EndFor \\
    \Return $(S \cup {N}_1 \cup {N}_2)$
\end{algorithmic}
\end{algorithm}
\vspace{-0.5cm}

\noindent \MkGreedy consists of two phases\ifFULL (see notation in Section~\ref{sec:notation})\else\fi. In Phase~(a), it runs \Greedy for the 
first $f$ fraction of the elements, but \emph{picks} each element \emph{selected} by
\Greedy into the final solution only with probability $(1-p)$, where $p>0$ is a constant. With the remaining probability $p$, it
\emph{marks} the element $e$, and behaves as if it had been
selected. The idea of marking some elements in Phase~(a) is that we hope to ``augment'' them in Phase~(b). To distinguish if an element is marked or picked,
the algorithm uses auxiliary random bits $\K$ that are unknown to the adversary. We
assume that $\K(e) \sim \text{Bern}(1-p)$ i.i.d. for all $e \in E$.

In Phase~(b), one needs to ensure that the augmentations of the marked elements
do not conflict with each other.  The crucial idea is to use the span of the
elements selected by \Greedy in Phase~(a) as a proxy to find two random disjoint OMI
subproblems.  The following Fact~\ref{swapbases} \ifFULL  (proof in Section~\ref{sec:missfacts}) \else \fi underlies this intuition. 
It states that given any independent set $S$, we can substitute it
 by any other independent set contained in the span of $S$. In
Lemma~\ref{lem:correctness} we use it to prove the correctness of \MkGreedy. \ifFULL \else Both Fact~\ref{swapbases} and Lemma~\ref{lem:correctness} are proved in the full version. \fi

\begin{fact}     \label{swapbases}
    Consider any matroid $\M$ and independent sets $A,B,C \in \M$  such that $A \subseteq \cl_{\M} (B)$ and 
    $B \cup C \in \M$. Then, $A \cup C \in \M$. 
\end{fact} 

\begin{lemma} \label{lem:correctness}
\MkGreedy outputs sets $S, {N}_1$, and ${N}_2$ such that 
\[(S \cup {N}_1 \cup {N}_2) \in \M_1 \cap \M_2. \]
\end{lemma}

\ifFULL
\begin{proof}
Observe that the outputs sets $S, {N}_1$, and ${N}_2$  of \MkGreedy satisfy the following for $i \in \{1,2\}$:
\begin{align}
    \Ni \in & ~\Mi / S \cap \Mni /T_f \qquad 
            &\left(\text{due to Line~\ref{SamAlg:claimprecon}} \right)  \label{eq:protect1}\\
    \Ni \subseteq & ~\cl_{\Mi/S}(T_f \setminus S) \qquad 
            &\left( \text{due to Line~\ref{algoB:span}}\right) \label{eq:protect2}
\end{align}
     From Property~(\ref{eq:protect1}) above we know
     ${\Nni} \in \Mi /T_f$, which implies
    ${\Nni} \cup (T_f\setminus S) \in \Mi /S$ because $S \subseteq T_f \in \Mi$.  
    Also, Property~(\ref{eq:protect2}) implies  $\Ni \subseteq \cl_{\Mi/S} (T_f\setminus S)$. Using  Fact \ref{swapbases}, we have  $ \big( {N}_1 \cup {N}_2 \big) \in \Mi / S$.
\end{proof}
\else
\fi

\subsubsection{Proof that \MkGreedy works for OMI:} \label{sec:matrAnalysis}
~\\
\noindent We know from Lemma~\ref{lemma:KMM} that $\G(f)$ is close to half for $\epsilon \ll f \ll 1$. In the following Lemma~\ref{lem:algoB}, we show that \MkGreedy (which returns $S\cup N_1\cup N_2$ by Lemma~\ref{lem:correctness}) gets an improvement over \Greedy. This completes the proof of Theorem~\ref{thm:onebitmat} to give $\gamma \geq 0.03$ 
 for $\epsilon = 0.001$, $f = 0.05 $, and $p =
0.33$. The rest of the section is devoted to proving the following lemma.


\begin{lemma} \label{lem:algoB}
    \MkGreedy outputs sets $S, {N}_1$, and ${N}_2$ such that
	\[ \E_{\pi,\K}[|S \cup {N}_1 \cup {N}_2|] \geq 
        (1-p)\, \G(f)\,|\OPT| + \frac{2p}{1+p} \left(1 -  \frac{2\epsilon}{f} -
            2\epsilon -f -\G(f) \right)\,|\OPT|. \]
\end{lemma}

\begin{proof}[Lemma~\ref{lem:algoB}]
We treat the sets $S \subseteq T_f, {N}_1,$ and ${N}_2$ as random sets depending on 
$\pi$ and $\K$. Since \MkGreedy ensures the sets are disjoint, 
\begin{align}
     \E_{\pi,\K}[ | S\cup {N}_1 \cup {N}_2 |] &= 
         \E_{\pi,\K}[|S|] +  \E_{\pi,\K}[|{N}_1| + |{N}_2|] \notag \\
         &\geq 
            (1-p)\,\G(f)\,|\OPT| + \E_{\pi,\K}[|{N}_1| + |{N}_2|]. \label{eq:setN}
\end{align}

Next, we lower bound
$\E_{\pi,\K}[|{N}_1| + |{N}_2|]$ by observing that for $i\in \{1,2\}$, $\Ni$ is
the result of running $\Greedy$ on  the following restricted set of elements.

\begin{definition} [Sets $\tE_i$] \label{defn:setsEI}
     For $i \in \{1,2\}$, we define $\tE_i$ 
    to be the set of elements $e$ that arrive in Phase~(b) and satisfy $e \in \cl_{i}(T_f)$ and $e
    \not\in \cl_{\thickbar{\imath}}(T_f)$. 
\end{definition}

\noindent It's easy to see that $\Ni$ is obtained by running $\Greedy$ 
on the matroids $\Mi /S$ and
$\Mni /T_f$ with respect to  elements in $\tE_i$, i.e. 
$\Ni =  \Greedy(\Mi/S,\Mni/T_f,\tE_i)$. To lower
bound $\E_{\pi,\K}[|{N}_1| + |{N}_2|]$, we use
the following Sampling Lemma \ifFULL (proved in Section~\ref{sec:misssamp}) \else (see full version) \fi that forms the core of our technical analysis. Intuitively, it says that if $S$ is a random subset of $T_f$ then for the obtained random OMI instance, with optimal solution of expected size $p\,|\tI|$,
\Greedy performs better than half-competitiveness even for adversarial arrival order of ground elements.


\begin{lemma}[{Sampling Lemma}]
\label{lem:matrmain} 
     Given matroids $\M_1, \M_2$ on ground set $E$, a set $T \in \M_1 \cap \M_2 $, and $\K(e) \sim  \text{Bern}(1-p)$ i.i.d. for all $e \in T$, we define set $S := \{e \mid e \in T \text{ and } \K(e) =    1\}.$ I.e., $S$ is a set achieved by dropping each element in $T$
    independently with probability $p$.
   For $i\in \{1,2\}$,  consider a set $\tE \subseteq \cl_i(T)$ and a set  ${\tI} \subseteq \tE $ satisfying  
    ${\tI} \in \Mi \cap (\Mni/T) $. Then 
    for any arrival order of the elements of $\tE$, we have
    \[ \E_{\K}[ \text{\Greedy}(\Mi/S, \Mni/T, \tE)] \geq \frac{1}{1+p} \left(p\,|{\tI} | \right) .\] 
\end{lemma}

To use the Sampling Lemma, in Claim~\ref{claim:propSetI} we argue that in expectation there exist disjoint sets ${\tI}_i
\subseteq \tE_i$ of ``large'' size that satisfy the preconditions of the Sampling Lemma (proof uses Hastiness Lemma and is deferred to \ifFULL Section~\ref{sec:omiDisjointSetsClaim}\else full version\fi).

\begin{myclaim} \label{claim:propSetI}
 If $\G(1) \leq \left(\frac12 + \epsilon \right)$ then for $i\in \{1,2\}$ $\exists$ disjoint sets $\tI_i \subseteq \tE_i$ s.t. 
 \begin{enumerate}[(i)]
        \item $\E_{\pi}\big[ |{\tI}_1| + |{\tI}_2| \big] 
                \geq 2\left(1 -  \frac{2\epsilon}{f}  -f -\G(f) \right)\,|\OPT| $.
       \item ${\tI}_i \in \Mi \cap (\Mni/T_f)$.  \label{eq:claimtildeI}
\end{enumerate}
\end{myclaim}

Finally, to finish the proof of Lemma~\ref{lem:algoB},  we use the sets ${\tI_i}$ from the above Claim~\ref{claim:propSetI} as ${\tI}$ and sets $\tE_i$ as $\tE$ in the Sampling Lemma~\ref{lem:matrmain}.
From Eq.~(\ref{eq:setN}) and Claim~\ref{claim:propSetI}, we get 
    \begin{align*}
         \E_{\pi,\K}[ | S\cup {N}_1 \cup {N}_2 |]& \geq (1-p)\, \G(f)\,|\OPT| 
         + \frac{p}{1+p}\, \E_{\pi} \big[ |{\tI}_1| + |{\tI}_2| \big] \\
         & \geq (1-p)\, \G(f)\,|\OPT| + \frac{2p}{1+p}\,\left(1 -  \frac{2\epsilon}{f}  -f -\G(f) \right)\,|\OPT| . 
    \end{align*}
\end{proof}

\ifFULL
\begin{proof}[Claim~\ref{claim:propSetI}]
Recall $\Phi_i(T_f^{\pi}) := \cl_i(T_f^{\pi}) \cap \OPT$.
Let $\I_i$ denote $\Phi_i(T_f^{\pi}) \setminus \Phi_{\thickbar{\imath}}(T_f^{\pi})$.
We construct sets $\tI_i$ by removing some elements from $\I_i$, which implies $\tI_i \in \Mi$ because $\I_i \in \Mi$ .
We first show that $|\I_1|+|\I_2|$ is large.
From the Hastiness Lemma~\ref{lemma:KMM}, we have 
\begin{align}
\E_{\pi} \left[ |\I_1| + |\I_2| \right] & = \E_{\pi} \left[ |\Phi_1(T_f^{\pi}) \cup \Phi_2(T_f^{\pi})| \right] - \E_{\pi} \left[ |\Phi_1(T_f^{\pi}) \cap \Phi_2(T_f^{\pi})| \right] \notag \\
&\geq \left(1 -  \frac{2\epsilon}{f}  \right)\,|\OPT| .  \label{eq:bothI}
\end{align}
 
Next, we ensure that $\tI_i \in \Mni/T_f$. Note that $\I_{\thickbar{\imath}} \subseteq \cl_{\thickbar{\imath}}(T_f )$.
    Let $X_{\thickbar{\imath}}$ denote a minimum subset of elements of $T_f $ such that
    $\cl_{\thickbar{\imath}}(X_{\thickbar{\imath}} \cup \I_{\thickbar{\imath}}) = \cl_{\thickbar{\imath}}(T_f)$. Since $\I_{\thickbar{\imath}}$ and $T_f$  are independent in $\Mni$, we have
    $|X_{\thickbar{\imath}}|  = |T_f | - |\I_{\thickbar{\imath}}|.$ 
    Now starting with $\left( \I_i \cup
    \I_{\thickbar{\imath}} \right) \in \Mni$, we add elements of $X_{\thickbar{\imath}}$ into it. We will remove at most $|X_{\thickbar{\imath}}|$ elements from $\I_i$  to get a set $\I_i'$ such that $(\I_i' \cup X_{\thickbar{\imath}} \cup \I_{\thickbar{\imath}})  \in    \Mni$ as $\left( X_{\thickbar{\imath}} \cup \I_{\thickbar{\imath}} \right) \in \Mni$. Using Fact~\ref{swapbases} and 
    $\cl_{\thickbar{\imath}}(X_{\thickbar{\imath}} \cup \I_{\thickbar{\imath}}) = \cl_{\thickbar{\imath}}(T_f )$, we also have 
    $\I_i' \cup T_f \in \Mni$. One can use a similar argument to obtain set  $\I_{\thickbar{\imath}}'$ and $X_i$ such that $\I_{\thickbar{\imath}}' \cup T_f \in \Mi$. Since  
    $\E_{\pi} \big[|X_i| \big] = \E_{\pi} \big[|T_f | - |\I_{i } | \big]$ , 
    \begin{align}
     \E_{\pi}[ |\I_1'| + |\I_2'|] &\geq \E_{\pi}[ |\I_1| + |\I_2| - |X_1| - |X_2| ] = 2 ~\E_{\pi} [ |\I_1| + |\I_2| -|T_f|]  \label{eq:bothIp}
     \end{align}
Finally, to ensure that  $\tI_i \subseteq \tE_i$, observe that any element $ e\in \I_i'$ already satisfies  $e \in \cl_{i}(T_f)$ and $e \not\in \cl_{\thickbar{\imath}}(T_f)$. To ensure that these elements also appear in Phase~(b), note that all elements of $\I_i'$ belong to $\OPT$. Hence, in expectation over $\pi$, at most $f\,|\OPT|$ of these elements can appear in Phase~(a). The remaining elements appear in Phase~(b). Thus, combining the following equation with Eq.~(\ref{eq:bothI}) and Eq.~(\ref{eq:bothIp}) completes the proof of Lemma~\ref{lem:algoB}
\[
 \E_{\pi}\big[ |{\tI}_1| + |{\tI}_2| \big] \geq  \E_{\pi}\big[ |{\I}_1'| + |{\I}_2'| \big]  - f\,|\OPT|   .
 \]
\end{proof}

\subsection{Existence of Large Disjoint Sets for Claim~\ref{claim:propSetI}} \label{sec:omiDisjointSetsClaim}

Finally, we prove the missing Claim~\ref{claim:propSetI} that in expectation there exist disjoint sets ${\tI}_i \subseteq \tE_i$ of ``large'' size that satisfy the preconditions of the Sampling Lemma

\begin{proof}[Claim~\ref{claim:propSetI}]
Recall $\Phi_i(T_f^{\pi}) := \cl_i(T_f^{\pi}) \cap \OPT$.
Let $\I_i$ denote $\Phi_i(T_f^{\pi}) \setminus \Phi_{\thickbar{\imath}}(T_f^{\pi})$.
We construct sets $\tI_i$ by removing some elements from $\I_i$, which implies $\tI_i \in \Mi$ because $\I_i \in \Mi$ .
We first show that $|\I_1|+|\I_2|$ is large.
From the Hastiness Lemma~\ref{lemma:KMM}, we have 
\begin{align}
\E_{\pi} \left[ |\I_1| + |\I_2| \right] & = \E_{\pi} \left[ |\Phi_1(T_f^{\pi}) \cup \Phi_2(T_f^{\pi})| \right] - \E_{\pi} \left[ |\Phi_1(T_f^{\pi}) \cap \Phi_2(T_f^{\pi})| \right] \notag \\
&\geq \left(1 -  \frac{2\epsilon}{f}  \right)\,|\OPT| .  \label{eq:bothI}
\end{align}
 
Next, we ensure that $\tI_i \in \Mni/T_f$. Note that $\I_{\thickbar{\imath}} \subseteq \cl_{\thickbar{\imath}}(T_f )$.
    Let $X_{\thickbar{\imath}}$ denote a minimum subset of elements of $T_f $ such that
    $\cl_{\thickbar{\imath}}(X_{\thickbar{\imath}} \cup \I_{\thickbar{\imath}}) = \cl_{\thickbar{\imath}}(T_f)$. Since $\I_{\thickbar{\imath}}$ and $T_f$  are independent in $\Mni$, we have
    $|X_{\thickbar{\imath}}|  = |T_f | - |\I_{\thickbar{\imath}}|.$ 
    Now starting with $\left( \I_i \cup
    \I_{\thickbar{\imath}} \right) \in \Mni$, we add elements of $X_{\thickbar{\imath}}$ into it. We will remove at most $|X_{\thickbar{\imath}}|$ elements from $\I_i$  to get a set $\I_i'$ such that $(\I_i' \cup X_{\thickbar{\imath}} \cup \I_{\thickbar{\imath}})  \in    \Mni$ as $\left( X_{\thickbar{\imath}} \cup \I_{\thickbar{\imath}} \right) \in \Mni$. Using Fact~\ref{swapbases} and 
    $\cl_{\thickbar{\imath}}(X_{\thickbar{\imath}} \cup \I_{\thickbar{\imath}}) = \cl_{\thickbar{\imath}}(T_f )$, we also have 
    $\I_i' \cup T_f \in \Mni$. One can use a similar argument to obtain set  $\I_{\thickbar{\imath}}'$ and $X_i$ such that $\I_{\thickbar{\imath}}' \cup T_f \in \Mi$. Since  
    $\E_{\pi} \big[|X_i| \big] = \E_{\pi} \big[|T_f | - |\I_{i } | \big]$ , 
    \begin{align}
     \E_{\pi}[ |\I_1'| + |\I_2'|] &\geq \E_{\pi}[ |\I_1| + |\I_2| - |X_1| - |X_2| ] = 2 ~\E_{\pi} [ |\I_1| + |\I_2| -|T_f|]  \label{eq:bothIp}
     \end{align}
Finally, to ensure that  $\tI_i \subseteq \tE_i$, observe that any element $ e\in \I_i'$ already satisfies  $e \in \cl_{i}(T_f)$ and $e \not\in \cl_{\thickbar{\imath}}(T_f)$. To ensure that these elements also appear in Phase~(b), note that all elements of $\I_i'$ belong to $\OPT$. Hence, in expectation over $\pi$, at most $f\,|\OPT|$ of these elements can appear in Phase~(a). The remaining elements appear in Phase~(b). Thus, combining the following equation with Eq.~(\ref{eq:bothI}) and Eq.~(\ref{eq:bothIp}) completes the proof
\[
 \E_{\pi}\big[ |{\tI}_1| + |{\tI}_2| \big] \geq  \E_{\pi}\big[ |{\I}_1'| + |{\I}_2'| \big]  - f\,|\OPT|.   
 \]
\end{proof}
\else
\fi

\ifFULL
\section{Sampling Lemma} \label{sec:misssamp}
We prove the lemma for $i=1$ as the other case is analogous.  

\subsection{Alternate View of the Sampling Lemma}

We prove the Sampling Lemma by first showing that $\Greedy( \M_1/S, \M_2/ T, \tE)$ produces
the same output as  algorithm \SampAlg (proof deferred to Section~\ref{sec:proofGreedyEquiv}).
    \begin{lemma} \label{clm:GreedyEquiv}
        Given a fixed $\K$ and assuming the elements of $\tE$ are presented in the same order, 
        the output of \SampAlg
        is the same as the output of $\Greedy(\M_1/S, \M_2/T,\tE).$ 
    \end{lemma}

The idea behind \SampAlg is to run  \Greedy, but postpone distinguishing
   between  the elements that are selected by \Greedy (set $T$) and picked by our algorithm (set $S$). This limits what an
    adversary can do while ordering  the elements of $\tE$. 
    Intuitively, the sets in \SampAlg denote the following: 
    \begin{itemize}
        \item $N'$ denotes the new elements to be added to the independent set. 
        \item $T'$ are the elements of $T$ for which we still haven't read the random bit $\K$. 
        \item $S'$ are the elements $e\in T$ for which we have read $\K$ and they turn out to be picked, i.e., $\K(e)=1$. 
    \end{itemize}
    
    \begin{algorithm}
    \caption{\SampAlg}
    \begin{algorithmic}[1]
    \Statex Input: $\M_1, \M_2, T$, and random bits $\K \in \{0,1\}^{|T|}$. 
    \State Initialize: $N',S'$ to $\emptyset$, and $T' = T$
    \For{each element $e \in \tE$} \label{SampAlg:FortE}
        \If{$ T \cup N' \cup e \in \M_2 $} \label{SampAlg:M2check}
                \State Let ${C} \gets C_1(S' \cup N' \cup T',e) \cap T'$ \label{SampAlg:circuit}	\Comment{\emph{Unread elements of the formed circuit}}
                \For{each element $f \in {C}$} \label{SampAlg:For}
                    \State $T' \gets T' \setminus f$ \label{SampAlg:Tminus}
                    \If {$\K(f) = 1$} \Comment{\emph{Auxiliary random bits $\K$}}
                        \State $S' \gets S' \cup  f $ \label{SampAlg:Sadd}  \Comment{\emph{Already picked elements}}
                    \Else
                        \State $N' \gets N' \cup e$ \label{SampAlg:Nadd}	\Comment{\emph{Newly picked elements}}
                        \State \textbf{Break}
                    \EndIf
                \EndFor
        \EndIf
    \EndFor\\
    \Return $N'$
    \end{algorithmic}
    \end{algorithm}


\subsection{Proof of the Sampling Lemma}

    By Lemma~\ref{clm:GreedyEquiv}, it suffices to prove that given the preconditions of
    the Sampling Lemma, \SampAlg produces 
    an output of expected size at least $\frac{p}{1+p}|\tI|.$      
    More precisely, we need to show that if $\K$ in \SampAlg is chosen as $\K(e) \sim  \text{Bern}(1-p)$ i.i.d. for all $e \in T$,
    we have $\E_{\K} [ |N'|] \geq \frac{p}{1+p} |\tI|.$

The main idea of the proof is to argue that before every iteration of the
for-loop in Line~\ref{SampAlg:FortE}, there are ``sufficient'' number of
elements that are still to arrive and can be added to our solution. 
To achieve this, we define a set $\I'$, which intuitively denotes the set of $\OPT$
elements that are still to arrive and can be added to the current solution. The properties
of $\I'$ are rigorously captured in Invariant~\ref{invariant}, where
 $\trE$ denotes the remaining elements of $\tE$ that are still to be considered in the for-loop. Due to Lemma~\ref{clm:GreedyEquiv}, this also denotes the elements of $\tE$ that are still to arrive for \Greedy. Starting with $\I' = \tI$ at the beginning of \SampAlg, we wish to maintain the following.
\begin{invariant}\label{invariant}
For given sets $S', N', T$, and $\trE \subseteq \tE$, we have set $\I'$ satisfying this invariant if 
\begin{align}
S' \cup N' \cup I' &\in \M_1  \label{inv1}   \\ 
            T \cup N' \cup I' &\in \M_2  \label{inv2} \\
            I' &\subseteq \trE \label{inv3}
\end{align}
\end{invariant}

As the algorithm \SampAlg progresses, set $\I'$ has to drop some of its elements so that it continues to satisfy Invariant~\ref{invariant}. These drops from $\I'$ are rigorously captured in Updates~\ref{updates}. Note that set $\I'$ and Updates~\ref{updates} are just for analysis purposes, and never appear in the actual algorithm. 
  Starting with
$\I' = \tI$ at the beginning of \SampAlg and satisfying
Invariant~\ref{invariant}, in
Claim~\ref{claim:invariant} we prove that Updates~\ref{updates} to $\I'$ ensure
that the invariant is always satisfied. This lets us use induction
to prove in Claim~\ref{claim:sublemma}  that Updates~\ref{updates} never drop too many elements from $\I'$ and \SampAlg returns an independent set of large size.

\begin{updates} \label{updates} We perform the following updates to $\I'$
    whenever \SampAlg reaches Line~\ref{SampAlg:Sadd} or
    Line~\ref{SampAlg:Nadd}. Claim~\ref{claim:invariant} shows that these
    updates are well-defined.  
    
\begin{itemize}
\item Line~\ref{SampAlg:Sadd}: If circuit $C_1(S' \cup N' \cup \I' \cup f)$ is
    non-empty then remove an element from $\I'$ belonging to $C_1(S' \cup N'
    \cup \I' \cup f)$ to break the circuit. 

\item Line~\ref{SampAlg:Nadd}: If circuit $C_1(S' \cup N' \cup \I' \cup e)$ is
    non-empty then remove an element from $\I'$ belonging to 
    $C_1(S' \cup N' \cup \I' \cup e)$ to break the circuit.  If 
    $C_2(T \cup N' \cup \I' \cup e)$ is non-empty then remove another element from $\I'$ belonging
    to $C_2(T \cup N' \cup \I' \cup e)$ to break the circuit. In the special
    case where $e \in \I'$, we remove $e$ from $\I'$.

\end{itemize}
\end{updates}

 The following claim (proof deferred to Section~\ref{sec:invariantproof}) shows that Updates~\ref{updates} maintain Invariant~\ref{invariant}.
\begin{myclaim}\label{claim:invariant}
    Given matroids $\M_1,\M_2$, a set $T \in \M_1 \cap \M_2$,  a set $\trE \subseteq \cl_1(T)$ 
    (denoting the set of remaining elements), 
    and $\K(e) \sim  \text{Bern}(1-p)$ i.i.d. for all $e \in T$, 
    suppose there exists a set $I'$ satisfying Invariant~\ref{invariant} at the beginning of some iteration of the for-loop in Line~\ref{SampAlg:FortE} of \SampAlg.  Then 
     \begin{enumerate}[(i)]
    \item Updates~\ref{updates} are well-defined. \label{prop1}
    \item Updates~\ref{updates} ensure that Invariant~\ref{invariant} hold 
        at the end of the iteration. \label{prop2}
    \end{enumerate}
\end{myclaim}

Finally, we use Invariant~\ref{invariant} to prove the main claim.
\begin{myclaim}\label{claim:sublemma}
    Given matroids $\M_1,\M_2$, a set $T \in \M_1 \cap \M_2$,  a set $\trE \subseteq \tE \subseteq \cl_1(T)$ 
    (denoting the set of remaining elements), 
    and $\K(e) \sim  \text{Bern}(1-p)$ i.i.d. for all $e \in T$, 
    suppose there exists a set $I'$ satisfying
    Invariant~\ref{invariant} at the beginning of some iteration of the for-loop of Line~\ref{SampAlg:FortE} in \SampAlg.  Then  the value of $N'$ at the end of \SampAlg satisfies
        \[ \E_{\K} [|N'|] \geq \frac{p}{1+p} |I'| \]
\end{myclaim}

\begin{proof}
    To prove the claim we use induction on $|\I'|$ where
    $\I' \subseteq \tE$.  WLOG we can assume that $e$ is the first
    element such that $C $ in Line~\ref{SampAlg:circuit} is non-empty.     
    Let $C = \{t_1,\dots, t_l\} $ where $l\geq 1$.
     For $j \in \{0,\dots, l-1\}$, define event $B_j$ as 
    $\K(t_1) = \K(t_2) = \dots = \K(t_j)=1$ and $\K(t_{j+1}) = 0$. 
    Also, define $\thickbar{B} $ as $ \K(t_1) = \dots \K(t_l) = 1$.

    \noindent \textbf{\emph{Base Case:} }
    Since $C$ is a non-empty circuit, we can assume that any element 
    $f \in C$ satisfies the condition $\K(f) = 0$ with probability $p$. 
    Hence, $ |N'| \geq 1$ with probability at least $p$, proving the required claim.

    \noindent \textbf{\emph{Induction Step:} }
    The events $B_0,\dots, B_{l-1}$, and $\thickbar{B}$ partition
    the entire probability space. 

    \noindent \emph{~~Case 1} (Event $B_j$) : Since applying the Updates~\ref{updates} preserves 
    Invariant~\ref{invariant} by Claim~\ref{claim:invariant}, 
    we can apply the induction hypothesis to the updated  set $I'$. Moreover. Updates~\ref{updates}
    remove at most $j+2$ elements from $I'$ in the event $B_j$. 
    Applying the Induction hypothesis, we can conclude that 
    $\E_{\K}[ |N'| \,\big\vert\,  B_j ] \geq  1+ \frac{p}{1+p} (|I'| - j-2)$. 

    \noindent \emph{~~Case 2} (Event $\thickbar{B}$): Since applying the Updates~\ref{updates} preserves
    Invariant~\ref{invariant} by Claim~\ref{claim:invariant}, 
    we can apply the induction hypothesis to the updated set $I'$. Moreover, Updates~\ref{updates}
    remove  $l$ elements from $\I'$ in the event $\thickbar{B}$. 
    Conditioned on the event $\thickbar{B}$ and applying the induction hypothesis 
    to the updated set $I'$, we can conclude 
     $\E_{\K} [ |N'| ] \geq \frac{p}{1+p} (|I'| - l) $.

     Combining both the cases, we have $  \E_{\K}[|N|]$ is at least
    \begin{align*}
       &  ~\sum_{j=0}^{l-1} \Pr[B_j] \cdot  \E_{\K}[ |N'| \,\big\vert\, B_j] +  
                    \Pr[\thickbar{B}] \cdot \E_{\thickbar{B}}[ |N| \,\big\vert\, \thickbar{B}] \\
        &\geq  \sum_{j=0}^{l-1} (1-p)^j \, p\left( 1 + \frac{p}{1+p}(|\I'| - 2
            -j) \right) + (1-p)^l \left( \frac{p}{1+p} (|\I'| - l)\right) \\
        & = \frac{p}{1+p}\,|\I'| \text{\qquad using $\sum_{j=0}^{l-1} j\,(1-p)^j =  -\frac{l\,(1-p)^l}{p} - \frac{(1-p)}{p^2}( (1-p)^l-1)$.}  
    \end{align*}
    \end{proof}
    To finish the proof of Lemma~\ref{lem:matrmain}, we start with 
    $\I' := {\tI}$, $T':=T$,  $N' :=\emptyset$, and $S':= \emptyset$ in
    Claim~\ref{claim:sublemma}. The preconditions hold true because
    $T\cup {\I} \in \M_2$, $T\in \M_1$, and $ {\I} \in \M_1$.

\subsection{Proof of the Alternate View of Sampling Lemma} \label{sec:proofGreedyEquiv}
We restate the lemma for convenience.

\noindent \textbf{Lemma~\ref{clm:GreedyEquiv}}. 
       \textit{ Given a fixed $\K$ and assuming the elements of $\tE$ are presented in the same order,         the output of \SampAlg
        is the same as the output of $\Greedy(\M_1/S, \M_2/T,\tE).$ }

Starting with $S'= \emptyset$, $N'= \emptyset$, and $T' = T$, we make some simple observations and prove a small claim before proving Lemma~\ref{clm:GreedyEquiv}.
    \begin{observation} \label{clm:SAinvariant2}
The for-loop defined in  Line~\ref{SampAlg:FortE} of \SampAlg maintains the following invariant
        \[ S \subseteq S' \cup T' \subseteq T \]
    \end{observation}
    \begin{proof}
        To show the first containment, observe that for each element if an $\K(e)=1$ then it simply
        moves from $T'$ to $S'$. Hence, all the elements of $S \subseteq S' \cup T'.$  
        To observe, the second containment, note that an element of $T'$ either moves into $S'$ or gets
        removed. Since $T'$ was initialized to $T$, the second containment follows. 
    \end{proof}

   \begin{observation} \label{clm:SAinvariant1}
The for-loop defined in Line~\ref{SampAlg:FortE} of \SampAlg maintains the following invariant 
        \[ S' \cup N' \cup T' \in \M_1 . \]
    \end{observation}
    \begin{proof}
        Since $T \in \M_1$ and $S'=T'=\emptyset$ at the beginning, we can conclude that this is correct at the
        beginning of \SampAlg. Now consider an iteration of  the for-loop defined in Line~\ref{SampAlg:FortE}.
 When an element $f$ is added to $S'$ in 
        Line~\ref{SampAlg:Sadd}, it must have belonged to $T'$, implying that $S' \cup N' \cup T'$ is unchanged.
        If an element $e$ is added to $N'$ (in Line~\ref{SampAlg:Nadd}) then we must remove an element $f$ from $T'$ 
        (due to Line~\ref{SampAlg:Tminus}),  which belonged to the unique circuit $C_1(S'\cup T' \cup N',e)$. Hence,
        $S' \cup N'\cup e  \cup (T'\setminus f)$  is still an independent set in $\M_1$. 
    \end{proof}

    \begin{myclaim} \label{clm:SAcircuits}
        For an element $e \in \tE$, if Line~\ref{SampAlg:circuit} of \SampAlg is reached  then 
        $C_1(S' \cup N' \cup T', e) $ is non-empty. 
    \end{myclaim}

    \begin{proof}
      We know $\tE \subseteq \cl_1(T)$. Moreover, $S' \cup T' \subseteq T \subseteq \cl_1(T)$ (using Observation~\ref{clm:SAinvariant2}). Hence, $S' \cup T' \cup \tE \subseteq \cl_1(T)$ implies 
\begin{align} \label{eq:boundedrank}
\rank_{\M_1}(S' \cup T' \cup \tE ) \leq |T|.
\end{align}
We prove the lemma by contradiction and assume circuit $C_1(S'\cup N'\cup T',e )$ is empty. Using Observation~\ref{clm:SAinvariant1}, this implies $\left(S'\cup N' \cup T' \cup e\right) \in \M_1$. Now,  $\rank_{\M_1}(S' \cup N' \cup T' \cup e) = |S' \cup N' \cup T'|+1 \leq \rank_{\M_1}(S' \cup T' \cup \tE ) \leq |T|$ using Eq.~(\ref{eq:boundedrank}). In the next paragraph, we show that the algorithm always maintains $|S' \cup N' \cup T'| = |T|$, which gives the contradiction $|T|+1 \leq |T|$.

 To prove $|S' \cup N' \cup T'| = |T|$, we note that the only time $T'$ decreases is in Line~\ref{SampAlg:Tminus}. In this case, we either add
        the dropped element to $S'$ in Line~\ref{SampAlg:Sadd} or a new element to $N'$ in Line~\ref{SampAlg:Nadd}. 
        Hence, the $|S' \cup N' \cup T'|$ is unchanged in the for-loop of Line ~\ref{SampAlg:FortE}. 
        Since we initialize $S'=N'=\emptyset$ and $T'=T$, we can conclude that this $|S'\cup N' \cup T'| =|T|$
        is maintained.
    \end{proof}

We now have the tools to prove Lemma~\ref{clm:GreedyEquiv}.
    \begin{proof}[Lemma~\ref{clm:GreedyEquiv}] 
        Let us assume the elements of $\tE$ are presented in order $e_1,\dots, e_t$ where $t=|\tE|$.  We will use induction on the following hypothesis. \\
        \textbf{\emph{Induction Hypothesis (I.H.):}} 
        After both algorithms have seen the first $k$ elements $e_1,\dots, e_k$, the set $N'$ in 
        \SampAlg is the same as the set of elements selected by $\Greedy(\M_1/S, \M_2/T,\tE)$.

        \noindent \textbf{\emph{Base Case:}}
        Initially, both algorithms have not selected any element. Hence, $N'= \emptyset$ 
        is the set of all elements selected by \Greedy.  

        \noindent \textbf{\emph{Induction Step:}}
        Suppose the I.H. is true for elements $e_1,\dots, e_{k-1}$ and we are considering element $e_k.$ 
        If $e_k$ does not satisfy $T \cup N' \cup e_k \in \M_2$, then it will also not 
        satisfy the same condition for \Greedy because $N'$ is the set selected by \Greedy (by I.H.) 
        and $N' \cup e \notin \M_2/T$. In this case we are done with the induction step. 	From now assume  $T \cup N' \cup e_k \in \M_2.$

        Suppose $e_k$ is added to $N'$ in \SampAlg,
        then we claim $\Greedy(\M_1/S, \M_2/T,\tE)$ will also select this element. The only location
        where $e_k$ could be added is Line~\ref{SampAlg:Nadd}. This occurs
        when we remove some appropriate element $f \in T'$ to ensure $S' \cup (T'\setminus f) \cup N' \cup e \in \M_1$. Furthermore $\K(f) = 0$ implies $f\notin S$.
        By Observation~\ref{clm:SAinvariant2}, set $S \subseteq S' \cup T' \setminus f$. Hence, $S' \cup (T'\setminus f) \cup N' \cup e \in \M_1$ implies $S \cup N' \cup e \in \M_1$ and  \Greedy will also select this element. 

        Next, suppose $e_k$ is not picked by the algorithm.  By Claim~\ref{clm:SAcircuits},
        we know that $C_1(S'\cup N' \cup T',e)$ is non-empty. In this case,
        every element $f \in C$ encountered
        in the for-loop of Line~\ref{SampAlg:For} must have had $\K(f)=1$. This implies that at the end
        of the for-loop of Line~\ref{SampAlg:For}, circuit $C_1(S' \cup N' \cup T', e) \subseteq S' \cup N'$. Since $S'\subseteq S$ 
        (by Observation~\ref{clm:SAinvariant2}), this gives $N' \cup e \notin \M_1/S.$ Hence, 
        \Greedy cannot select element $e_k.$ 
    \end{proof}

\subsection{Proof that the Updates are valid} \label{sec:invariantproof}

In this section we prove Claim~\ref{claim:invariant} by showing that Updates~\ref{updates} are well-defined and maintain Invariant~\ref{invariant}.

\begin{proof}[Claim~\ref{claim:invariant}]
    Since Invariant~\ref{invariant} holds before entering into the for-loop in
    Line~\ref{SampAlg:FortE}, we prove this claim by showing that after one
    iteration of the for-loop, i.e., after arrival of an element $e$,
    Properties~(\ref{prop1}) and (\ref{prop2}) hold.
        
    We first show that the properties hold if the set $C$ in
    Line~\ref{SampAlg:circuit} is empty. Since in this case we do not perform
    any updates to sets $S',N', \I'$, and $T'$, Invariant~\ref{inv1},
    Invariant~\ref{inv2}, and well-definedness trivially hold. To prove
    Invariant~(\ref{inv3}), we need to show  $e \not\in \I'$. This is true
    because by Claim~\ref{clm:SAcircuits}  element $e$  forms a circuit in
    $C_1(S' \cup N' \cup T',e)$,  and by Invariant~(\ref{inv1}) we know $S'
    \cup N' \cup I' \in \M_1$. Hence, the circuit $C_1(S' \cup N' \cup T',e)$
    contains some element of $T'$, which gives the contradiction that $C$ is non-empty.

    Now WLOG, we can assume that element $e$ forms a non-empty set $C$
    in  Line~\ref{SampAlg:circuit}. We prove
    Property~(\ref{prop1}), Invariant~(\ref{inv1}), and Invariant~(\ref{inv2})
    by showing that they hold after any iteration of the for-loop
    of Line~\ref{SampAlg:For}.  Note that  sets $S', N'$, and $\I'$ can only
    change in Lines~\ref{SampAlg:Sadd} or \ref{SampAlg:Nadd} of the for-loop.
    We prove the claim for both these cases. 

    \noindent \textbf{\emph{~~Case 1}} (Line~\ref{SampAlg:Sadd}):  Since $f$ belonged to $T'$,
     from Observation~\ref{clm:SAinvariant1} we
     know $\left(S' \cup N' \cup f\right) \in \M_1$. Now using Invariant~(\ref{inv1}) (which
     holds before the iteration), we can
     deduce that $C_1(S' \cup N' \cup I', f) \cap I'$ is non-empty and the update
     is well-defined. Invariant~(\ref{inv1}) holds because the update breaks any
     circuit in $S' \cup N' \cup \I'$ in $\M_1$. Since $T$ and $N'$ are unchanged
     and $\I'$ only gets smaller, Invariant~(\ref{inv2}) holds trivially.  

     \noindent \textbf{\emph{~~Case 2}} (Line~\ref{SampAlg:Nadd}): Since we are adding $e$ to $N'$, it must
     be the case that $S' \cup N' \cup e \in \M_1$ (by
     Lemma~\ref{clm:GreedyEquiv}). If $C_1(S' \cup N' \cup I' \cup e)$ is non-empty
     then $C_1(S' \cup N' \cup \I' \cup e) \cap I'$ must be non-empty. Moreover, by
     Line~\ref{SampAlg:M2check}, we know  that $T
     \cup N' \cup e \in \M_2$. Hence, if $C_2 (T \cup N' \cup I' \cup e)$ is
     non-empty then $C_2(T \cup N' \cup I' \cup e) \cap I'$ must be non-empty. Both
     of them together prove the the update is
     well-defined in this case. Invariant~(\ref{inv1}) and Invariant~(\ref{inv2})
     hold because Updates~\ref{updates} break any circuit $C_1(S' \cup N' \cup \I'
     \cup e)$ and $C_2(T \cup N' \cup \I' \cup e)$.

    Finally, to finish the proof of this claim, we show that
    Invariant~(\ref{inv3}) also holds at the end of every iteration of the for-loop
    of Line~\ref{SampAlg:FortE}.  If $e \not\in \I'$ then
    Invariant~(\ref{inv3}) trivially holds as $\trE$
    looses element $e$, and $\I' \subseteq \trE \setminus e$. Now suppose $e \in \I'$.
    Here we consider two cases. 
    
    \noindent \textbf{\emph{~~Case 1}} ($e$ is added to $N'$): Here \SampAlg reaches
    Line~\ref{SampAlg:Nadd} and the special case of Update~\ref{updates} ensures
    that $e$ is removed from $\I'$. Hence $\I' \subseteq \trE$. 

    \noindent \textbf{\emph{~~Case 2}} ($e$ is not added to $N'$): From the above proof, we know
    that Invariants~(\ref{inv1}) and (\ref{inv2}) are preserved at the end of
    this iteration. We prove by contradiction
    and assume that $e \in I'$ at the end of this iteration. Since $e \notin N'$, all the elements
    of $C$ in Line~\ref{SampAlg:circuit} are added to $S'$ by the end of this iteration. Hence, 
    the entire circuit in Line~\ref{SampAlg:circuit} (which is non-empty by
    Claim~\ref{clm:SAcircuits}) is contained in $S' \cup N' \cup e$ at the end
    of the iteration. Since $e \in I'$, this implies that   
    $S' \cup N' \cup I'$ is not independent.  This is a contradiction as 
    Invariant~(\ref{inv1}) is violated.
\end{proof}

\else

\fi


\ifFULL
\section{Beating Half for General Graphs} \label{sec:generalgraphs}

\begin{figure}
	\centering
	\begin{tikzpicture} [thin,scale=1.5]
	\draw (0,1.5) ellipse (0.5cm and 2cm)[fill=black!0];
	\draw (2,1.5) ellipse (0.5cm and 2cm)[fill=black!0];
	\foreach \y in {0,...,2}
	{
	\draw [-, very thick] (-0.25 , 0.5*\y) to node[auto]{$ $} (0.3 , 0.5*\y + 0.2);
	\draw [-, densely dashdotted, thick] (-0.25, 0.5*\y) to node[auto]{$ $} (2, 0.5*\y);
	\draw [-, densely dashdotted, thick] (0.3, 0.5*\y + 0.2) to node[auto]{$ $} (2, 0.5*\y + 0.2);
	\node at (-0.25 , 0.5*\y) [graphnode]{};
	\node at (0.3 , 0.5*\y + 0.2) [graphnode]{};
	\node at (2 , 0.5*\y) [graphnode]{};
	\node at (2 , 0.5*\y + 0.2) [graphnode]{};
	}
	\foreach \y in {3,...,5}
	{
	\draw [-, dashed] (-0.25 , 0.5*\y) to node[auto]{$ $} (0.3 , 0.5*\y + 0.2);
	\draw [-, densely dashdotted, thick] (-0.25, 0.5*\y) to node[auto]{$ $} (2, 0.5*\y);
	\draw [-, densely dashdotted, thick] (0.3, 0.5*\y + 0.2) to node[auto]{$ $} (2, 0.5*\y + 0.2);
	\node at (-0.25 , 0.5*\y) [graphnode]{};
	\node at (0.3 , 0.5*\y + 0.2) [graphnode]{};
	\node at (2 , 0.5*\y) [graphnode]{};
	\node at (2 , 0.5*\y + 0.2) [graphnode]{};
	}
	\node at (-1,1.5) {$U$};
	\node at (3,1.5) {$V$};

	\node at (-.25,2.7) {$u_2$};
	\node at (0.23,2.9) {$u_1$};
	\node at (2.25,2.55) {$v_2$};
	\node at (2.25 ,2.8) {$v_1$};
\end{tikzpicture}
    \caption[margin=5cm]{ $U$ denotes the set of vertices matched by \Greedy in Phase~(a) and $V$ denotes the remaining vertices of $G$. Solid edges within $U$ denote the picked edges and dashed edges within $U$ denote the marked ones. Dashed edges from $U$ to $V$ denote the $\OPT$ edges.} 
        \label{figOME}
\end{figure}
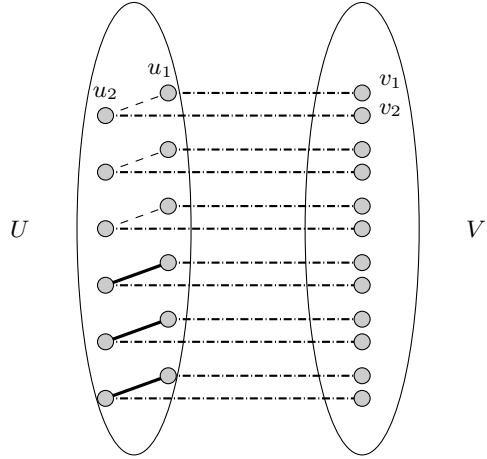		

\noindent \textbf{Theorem~\ref{thm:generalgraphs}.}
\emph{In the random edge arrival model, the  online matching problem for general graphs has
     a $(\frac{1}{2} + \delta')$-competitive randomized algorithm, where $\delta' >0$ is a constant.
}
\begin{proof} [Proof overview]
Let $G$ be the arrival graph with edge set $E$.
Using the same idea as Lemma~\ref{lem:reduc}, we can again focus on graphs where \Greedy has a competitive ratio of  at most $\left(\frac12 + \epsilon \right)$ for any constant $\epsilon>0$.
We construct  a two-phase algorithm that uses the algorithm from Theorem~\ref{thm:bipmatching} as a subroutine. In Phase~(a), we run \Greedy; however, each edge selected by \Greedy is picked only with probability $(1-p)$. With probability $p$, we mark it along with its vertices and behave as if it has been matched for the rest of Phase~(a). Since the hastiness property (Lemma~\ref{lemma:KMMmatch}) is also true for general graphs, in expectation we  pick $(1-p)\left( \frac12 - O(\frac{\epsilon}{f})\right)\,|\OPT|$ edges and mark $p\left( \frac12 - O(\frac{\epsilon}{f})\right)\,|\OPT|$ edges in Phase~(a). Now we need  to ensure that in expectation $(1+\gamma)$ edges, for some constant $\gamma>0$, are picked per marked edge in Phase~(b).

Let $T_f$ denote the set of edges selected by \Greedy in Phase~(a), i.e., both picked and marked edges.  Let $U$ denote the set of vertices matched in $T_f$ and $V$ denote the remaining set of vertices of $G$. Using the following simple Fact~\ref{lemma:KMMsimple} and Lemma~\ref{lemma:KMMmatch}, we can argue that $\left( 1 - O(\frac{\epsilon}{f}) \right)$ $\OPT$ edges go from a vertex in $U$ to a vertex in $V$ in graph $G$. 

\begin{fact}[Lemma~$1$ in~\cite{KMM-APPROX12}] \label{lemma:KMMsimple}
    Consider a maximal matching $T$ of graph $G$ such that 
    $|T| \leq \left(\frac{1}{2} + \epsilon \right) |\OPT|$ for some $\epsilon \geq 0$.  Then
    $G$ contains at least $\left(\frac{1}{2} - 3\epsilon \right) |\OPT|$ vertex
    disjoint 3-augmenting paths with respect to $T$. 
\end{fact}

\noindent Moreover, in expectation at most $f$ fraction of these $(U,V)$ $\OPT$ edges can appear in Phase~(a). Thus, setting $\epsilon \ll f \ll 1$ gives that most of the $\OPT$ edges, i.e., $\left( 1 - O(\frac{\epsilon}{f}) -f \right)$ fraction, appear in Phase~(b). This implies that most of the marked edges contain two 3-augmentation edges as shown in Figure~\ref{figOME}. 

Now
consider a marked edge $(u_1,u_2)$ with $(u_1,v_1)$ and $(u_2,v_2)$ denoting its 3-augmentations. In comparison to bipartite graphs, the new concern in general graphs is that there might be an edge between $u_1$ and $v_2$ as triangles are possible in non-bipartite graphs. Hence, the Sampling Lemma~\ref{lemma:samp} cannot be directly applied here. However, we are only interested in the bipartite graph  between vertices $U$ and $V$. Therefore, in Phase~(b), we run the algorithm from Theorem~\ref{thm:bipmatching} for bipartite graphs restricted to $(U,V)$ edges. For sufficiently small values of constants $\epsilon$  and $f$, the constant $\delta$ gain in Theorem~\ref{thm:bipmatching} is sufficient to obtain a constant $\delta'$ gain for this theorem.
\end{proof}

\else
\fi


\ifFULL
\subsubsection*{Acknowledgments}
We are grateful to Anupam Gupta for several  discussions on the problem.  We are thankful to Ashwinkumar Badanidiyuru for pointing us connections between online and offline matroid intersection. We thank Anupam Gupta, Manuel Blum, Bernhard Haeupler,  Deeparnab Chakrabarty, Euiwoong Lee, and David Wajc for feedback on an earlier draft of the paper.
\else
\fi
{
\bibliographystyle{abbrv}
\bibliography{bib}
}

\ifFULL
\appendix

 \section{Notation} \label{sec:notation}

\begin{table}[htbp!]
\begin{tabular}{r c p{14cm} }

\multicolumn{3}{c}{General Notation}\\
\multicolumn{3}{c}{}\\
$\M_{i}$ &  & Matroid indexed by $i$\\
$ A \in \M$ & & Subset $A$ is an independent set in the matroid $\M$ \\
$T \cup e$ & & Short form for notation $T \cup \{ e \} $\\
$\rank_{\M}$ & & The rank function defined by matroid $\M$ \\
$\thickbar{\imath}$ & & Denotes the index $3-i$ \\
$\M_1 \cap \M_2$ & & The set of subsets that are independent in both matroids $\M_1$ and $\M_2$  \\
$\M / T$ & & The matroid resulting from  contracting subset $T$ in matroid $\M$ \\
$\cl_i(T)$ & & $\{ e \mid (e \in E) \wedge (\rank_{\M_i}(T\cup \{e\}) = \rank_{\M_i}(T) \}$ \\
$\C_i(T \cup e)$ & & The unique circuit formed by $T \cup \{e\} $ in matroid $\M_i$. This  is undefined when $T$ is
not an independent set and $e \notin \cl_i(T)$.  \\
$ E$ & & The set of ground elements common to the matroids $\M_1$ and $\M_2$ \\
$ \pi$ & & A permutation on the set $E$ \\
$\OPT$ & & A fixed maximum independent set in the intersection of $\M_1 \cap \M_2$ \\
$\G(f)$ & & $\E_{\pi}[ | T_f|] / |\OPT|$ \\
  
\multicolumn{3}{c}{}\\
\multicolumn{3}{c}{}\\
\multicolumn{3}{c}{Notation used by \MkGreedy in Section~\ref{sec:matrOutline}}\\
\multicolumn{3}{c}{}\\
$\K$ & & The set of random bits used in the algorithm. For each $e \in E$, we have $\K(e) \sim Bern(1-p)$ \\
\emph{selecting}& & The element is chosen by  $\Greedy$ in Phase~(a)  \\
\emph{picking} & & The element is chosen  by \MkGreedy in the final solution \\
\emph{marking} & & The element is chosen by $\Greedy$ in Phase~(a) but the algorithm does not pick it \\
$T_f$ & & The set of elements selected by $\Greedy$  in Phase~(a) \\
$S$ & & The set of elements picked by \MkGreedy in Phase~(a)  \\
$N_i$ & & The set of elements belonging to $\M_i /S \cap \Mni /T$ picked by \MkGreedy in Phase~(b)  \\
\end{tabular}
\caption{Table of Notation}
\label{tab:TableOfNotation}
\end{table}


\section{Miscellaneous Results} \label{appn:otherResults}
\subsection{\Greedy Beats Half on Almost Regular Graphs}
\begin{theorem}
For online matching in random edge arrival model, \Greedy has a competitive ratio of at least $(1-\frac1e)$ on any $d$-regular graph.
    \end{theorem}
\begin{proof}
Consider a vertex $v$,  and let $u_1,u_2,\dots,u_d$ be its neighbours. The probability that $(u_1,v)$ is the first to occur amongst all the edges of $u_1$ is exactly $\frac1d$. If this occurs, then we know that vertex $v$ will be surely matched. Thus, the probability that $v$ is not matched by the end of the algorithm is at most $(1-\frac1d)^d \leq \frac1e$. This means that each vertex is matched with probability at least $1-\frac1e$, leading to the stated theorem.
\end{proof}

The same analysis also extends to graphs that are almost regular, i.e., graphs with vertex degrees  
between $d\,(1 \pm \epsilon)$, for any small constant $\epsilon$.


\subsection{\Greedy Cannot Always Beat Half for Bipartite Graphs}
Dyer and Frieze~\cite{DF-RANDOM91} showed a general graph\footnote{This graph is popularly known as a \emph{bomb graph}. It is obtained by adding a new vertex and edge adjacent to each vertex of a complete graph.} for which \Greedy is half competitive. Inspired from their construction, we give the following bipartite graph for which \Greedy is  half competitive.

\begin{definition}[\textbf{Thick-\thickz graph}] \label{defn:thickz}
Let graph Thick-\thickz$:= ( (U_1 \cup U_2) \cup (V_1 \cup V_2) \}, E)$ be a bipartite graph with $|U_1|=|V_1|$ and $|U_2|=|V_2|$. The edge set $E$ consists of
the union of a perfect matching between $U_i$ and $V_i$ for $i\in \{ 1,2 \}$ and a complete bipartite graph between $U_2$ and $V_1$. If additionally $|U_1| = |V_2|$, we call the graph a balanced Thick-\thickz.
\end{definition}

\begin{lemma}
When the edges of a balanced Thick-\thickz are revealed one-by-one in a random order to  \Greedy then in expectation it produces a matching of size  $\left(\frac{1}{2} + o(1)\right)|\OPT|$.
\end{lemma}
\begin{proof}
We note that after an edge is picked by \Greedy, both the end points of the
edge do not participate later in the algorithm. Hence, at any instance during
the execution of \Greedy, the participating graph is still a Thick-\thickz
graph $( (U_1' \cup U_2') \cup (V_1' \cup V_2') \}, E')$, where 
$U_i' \subseteq U_i$ and $V_i' \subseteq V_i$ for $i\in \{1,2\}$.

We can view the choices made by \Greedy as being done in time steps, where \Greedy chooses
one edge at each time step.  At each time step, at least one of $U_1$ or $U_2$ decrease by $1$, and 
\Greedy halts when $ |U_1'| =  |U_2'| = 0$. 
Let $t$  be the random variable indicating  the first time step during the execution of \Greedy  when
$\min\{ |U_1'|, |U_2'|\} = n^{2/3}$. Let $a,b$ be the random variables denoting $a := |U_1'|=|V_1'| $ 
and $b := |U_2'|=|V_2'| $ at time $t$.  Let $O_1$ denote the number of edges of  $\OPT$  chosen by \Greedy
before time $t$ and let $O_2$ denote the number of edges of $\OPT$ chosen after time $t$.

We observe that the matching produced by \Greedy is of size $\frac{n}{2}+ |O_1|
+ |O_2|$.  Observe  $| \big(|U_1'| - |U_2'|\big) | $ changes only when \Greedy chooses an edge from $\OPT$,  
implying that we can bound $|a-b| \leq |O_1|$. Since $O_2$ is bounded by $|U_1'| + |U_2'|$ at time $t$, we can say 
\[ |O_2| \leq a + b = 2\,\min\{a,b\} +
    |a-b|  \leq 2\, n^{2/3} + |O_1|. \]

Next, to bound $|O_1|$, we note that before time $t$  the probability of an
edge picked by \Greedy being from $\OPT$ is at most 
$\frac{2n}{n^{2/3} \cdot n^{2/3}} = \frac{2}{n^{1/3}} $. Since \Greedy picks at most $n$ edges
before time $t$, we have $\E[|O_1|] \leq \frac{2n}{n^{1/3}} = 2\,n^{{2/3}} $. 
This proves that expected size of the matching chosen by  \Greedy is 
$\frac{n}{2}+ \E[|O_1| + |O_2|] \leq
\frac{n}{2} + 2\, n^{2/3} + 2\, \E[|O_1|] \leq \frac{n}{2} + 6\,n^{2/3}$. 

%
\end{proof}


\subsection{Limitations on any OBME Algorithm}
\begin{lemma}
No randomized algorithm can achieve a competitive ratio
greater than $\frac56 \sim 0.833$  for online bipartite matching in random edge arrival
model when the graph is a balanced Thick-\thickz  with $n=1$. This is true,
even the adversary knows the graph and can identify one vertex which has degree $2$. 
\end{lemma}
\begin{proof}
The optimum offline matching size is two.
However, no randomized online algorithm, (even one which knows the input graph), 
can obtain more than $\frac53$ edges in expectation over the random
edge order. 
To see this, let $p$ denote the probability that the algorithm picks 
the first edge it sees.

\emph{Case 1}: The first edge is from the optimal matching (i.e.~the first edge is 
of the form $(u_i,v_i)$ for $i \in \{1,2\}$).  In this case,
the algorithm will achieve the optimal value $2$ with probability $p$. 
If it skips one of these edges, it will retain at most $1$ edge in the remaining graph.

\emph{Case 2}: The first edge is not from the optimal matching (i.e. the first edge
is $(u_1, v_2)$). In this case the algorithm will achieve a value of at most
$1\cdot p+ 2\cdot(1-p)$. 

Since Case~$1$ occurs with probability $\frac23$ and Case~$2$ occurs with probability $\frac13$,
the expected value of the algorithm is $\frac53p + \frac43(1-p) \leq \frac53$. 
\end{proof}

\begin{figure}
	\centering
	\begin{tikzpicture} [thin,scale=1.5]

		\draw [-, line width=0.5mm, red] (0,0) to node {$ $} (1.5,-0.5);
		
		\draw [-, densely dashdotted, thick] (0,0.2) to node[auto]{$ $} (1.5,0.2);
		\draw [-, densely dashdotted, thick] (0,-0) to node[auto]{$ $} (1.5,0);

		\draw [-, densely dashdotted, thick] (0,-0.5) to node[auto]{$ $} (1.5, -0.5);
		\draw [-, densely dashdotted, thick] (0,-0.7) to node[auto]{$ $} (1.5, -0.7);
		
		\draw [-,   thick, blue] (0,-0.7) to node[auto]{$ $} (1.5,-0.5);
		\draw [-,   thick, blue] (0,0) to node[auto]{$ $} (1.5,0.2);		

		
        \node[graphnode] at (0,0){$ $};
        \node[graphnode] at (1.5,0.0){$ $};
        \node[graphnode] at (0,0.2){$ $};
        \node[graphnode] at (1.5,0.2){$ $};
        \node[graphnode] at (0,-0.5){$ $};
        \node[graphnode] at (1.5,-0.5){$ $};
		 \node[graphnode] at (0,-0.7){$ $};
        \node[graphnode] at (1.5,-0.7){$ $};
		
		\node at (2.5,0.1) {$Z_1$};
		\node at (2.5,-0.6) {$Z_2$};

	\end{tikzpicture}
    \caption[margin=5cm]{ The above example is a conjunction of two 
        Thick-\thickz graphs ($Z_1$ and $Z_2$) by a single edge (the thick red edge).
      Note that for a Thick-\thickz graph even knowing the degree $2$ vertex does
    not allow any algorithm to achieve more than $\frac53$ edges in expectation. }
        \label{figZ2}
\end{figure}
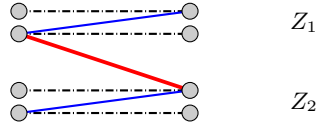

\begin{lemma}
No randomized algorithm can achieve a competitive ratio
greater than $\frac{69}{84} \sim 0.821 $ for online bipartite matching in random edge arrival
model.
\end{lemma}

\begin{proof}
    Our instance corresponds to the case where we take two copies of balanced
    Thick-\thickz graph joined by a single edge (see Figure~\ref{figZ2}). 
    The input is some permutation of the graphs (where the vertices or edges
    may be permuted and $U$ and $V$  may be swapped). We show by case analysis that
    no algorithm can achieve a competitive ratio better than $\frac{69}{84} <
    \frac56$.  Intuitively, the addition of the single edge only  hurts any algorithm
    without compromising the independence between the two instances.     

    Let $p$ be the probability that the algorithm picks the
    first edge. Consider the  following cases based on Figure~\ref{figZ2}:

    \emph{Case 1}: Suppose the first edge is the thick red edge. 
    This occurs with probability $\frac17$. If the algorithm picks this edge (which happens
    with probability $p$), then the optimal value in the remaining graph is
    $3$. Otherwise, it can get at most $2 \cdot \frac{5}{3}$  as the two Thick-\thickz
    graphs are disjoint and we can use the previous lemma. Hence the expected
    outcome is $\frac17 \big( p \cdot 3 + (1-p) \cdot \frac{10}{3} \big)$.

    \emph{Case 2}: Suppose the first edge is a blue edge, this occurs with probability
    $\frac27$. If the algorithm chooses this edge, then we can get value of $1$. Since 
    this affords no information about the second $Z$, the best an algorithm can do is $\frac53$. 
    Hence the expected solution is $\frac27 \big( p (1 + \frac53) + (1-p) (2 + \frac53) \big)$.

    \emph{Case 3}: Suppose the first edge is a black edge. This occurs with probability 
    $\frac47$. If the algorithm chooses the first edge, then we can get value of $2$ in this copy
    of the Thick-\thickz. However, still the algorithms gets at most $\frac53$ in the remaining copy
    of Thick-\thickz. Hence the expected cost of the solution is 
    $\frac47 \big(p(2 + \frac53) + (1-p) (1+\frac53) \big)$

    Adding these cases together, we get that expected solution has value at most $\frac{64+5p}{21}$.
    Since the optimal solution is $4$, this gives an upper bound of $\frac{69}{84}$. 

\end{proof}

\subsection{When Size of the Ground Set is Unknown} \label{sec:unknownm}

\begin{theorem} 
For any constant $\epsilon>0$, any randomized algorithm $\mathcal{A}$ 
that does not know the number of edges to arrive  has a competitive ratio $\alpha \leq \frac{2}{3} + \epsilon$ for online bipartite matching in random edge arrival model.
\end{theorem}
\begin{proof}
    To prove this theorem, we show that for any $\epsilon >0$ there exists
    an instance where $\mathcal{A}$ is less than $\frac{2}{3}+\epsilon$-competitive. 

    Since $\mathcal{A}$ does not know the number of edges to arrive, it must
    maintain an $\alpha$ approximation in expectation after the arrival of
    every edge. This is because $\mathcal{A}$ does not know if the current edge
    will be the last edge. 

    Consider the instance given by the graph balanced Thick-\thickz (see
    Definition~\ref{defn:thickz}) where the size of the $|U_1|=|V_1|=N$ will be
    set later. 
    Consider  a random permutation $\pi$ on the set of all edges and note that
    each edge $e$appears in the first $T$ edges with probability $\frac{T}{N^2 +
    2N} $, where $T = 4 (N+2) \log N$. The previous probability is at least
    $\frac{4\log N}{N}$. Let $G_T$ denote the set of edges from the perfect
    matching between $U_i$ and $V_i$ that appear in the first $T$ edges. Let $B_T$
    denote the set of edges from $U_2$ to $V_1$ that appear in the first $T$ edges.
    By linearity of expectation, we can say 
    $\E[|G_T|] \leq 8 \log N$ and $\E[|B_T|] \leq 4 N \log N$. 

    Let $\OPT_T$ denote the expected size of the maximum
    matching on the graph induced by the first $T$ edges. 
    \begin{myclaim}
        $ N(1-\epsilon) \leq \E_{\pi}[|\OPT_T|] $
    \end{myclaim}
    \begin{proof}
        Consider the graph induced between $U_2$ and $V_1$ in the first $T$ edges. Since
        any particular edge occurs with probability $\frac{4 \log N }{N}$ and the edges are
        negatively correlated, we can conclude that 
        \begin{align*}
        \Pr[ \exists \text{a perfect matching between $U_2$ and $V_1$ in the first $T$ edges}] &\geq \\
   & \hspace{-4cm}         \Pr[\exists \text{ a perfect matching in }\mathcal{G}_{N,N,\frac{4 \log N}{N}}] .
         \end{align*}
        By a result of Erdos and Renyi (see~\cite{EF-64}), we know that 
        \[\lim_{N \to \infty} \Pr[ \exists \text{a perfect matching in } \mathcal{G}_{N,N,\frac{4 \log N}{N}}] =1 \]
        Hence, we can choose an $N$ such that the above probability is at least
        $1-\epsilon$. Thus we can conlude that $\E[\OPT_T] \geq N(1-\epsilon) $.
    \end{proof}
    
    Let $M_{\OPT}$ denote the expected number of edges
    picked by $\mathcal{A}$ that belong to the perfect matching between
    $U_i$ and $V_i$ (for $i=1,2$) at time $T$.  Similarly, let $M_{Rest} $ denote the
    expected number of edges between $U_2$ and $V_1$ chosen by $\mathcal{A}$. 

    Since $\mathcal{A}$ must maintain an $\alpha$ approximation, we can say 
    $M_{\OPT} + M_{Rest} \geq \alpha( 1 - \epsilon)N$. 
    Since $M_{\OPT} \leq \E[|G_T|]= 8\log N \leq \alpha \epsilon  N$, we can say 
     \begin{align}
        M_{Rest} \geq (\alpha - 2\epsilon) N \label{eqn:mrest}
    \end{align}  However, every edge chosen from $M_{Rest}$
     decreases the  value of the optimal algorithm by one. Let $F$ be the expected size 
     of the matching chosen by the algorithm.  We know that $\alpha \cdot 2N
    \leq F \leq 2N - M_{Rest}$. Substituting into Eq.~(\ref{eqn:mrest}) and
    dividing by $2\,N$ ,  we get $\alpha \leq \frac{2}{3} + \epsilon$. 
\end{proof}


\ifFULL
\section{Facts}\label{sec:missfacts}
\noindent \textbf{Fact~\ref{fact:matching}.}
    \begin{align*}
        |T_1^{\pi}| &\geq  \frac12 \left( |\OPT| + \sum_{e \in \OPT}  
                \mathbf{1}[\text{Both ends of $e$ matched in $T_f^{\pi}$}] \right) \text{ and }\\
        |T_1^{\pi}| &\geq |T_f^{\pi}| + \frac12 \sum_{e \in \OPT}  
            \mathbf{1}[\text{Both ends of $e$ unmatched in $T_f^{\pi}$}]. 
    \end{align*}
  
\begin{proof}
     We start by counting the vertices matched in $T_1^{\pi}$, 
    \[ 2\, |T_1^{\pi}| \geq 2\, \sum_{e \in \OPT}   \mathbf{1}[\text{Both ends of $e$ matched in $T_1^{\pi}$}]
                    + \sum_{e \in \OPT}  \mathbf{1}[\text{Exactly one end of $e$ matched in $T_1^{\pi}$}] \]
    Since $T_1^{\pi}$ is a maximal set, 
    \[ |OPT| = \sum_{e \in \OPT}  \mathbf{1}[\text{Exactly one end of $e$ matched in $T_1^{\pi}$}]  
                    +\sum_{e \in \OPT}  \mathbf{1}[\text{Both ends of $e$ matched in $T_1^{\pi}$}] \]
    Combining the previous two statements and the fact that $T_f^{\pi} \subseteq T_1^{\pi}$,
    \[ |T_1^{\pi}| \geq  \frac12 \left( |\OPT| 
             + \sum_{e \in \OPT}  \mathbf{1}[\text{Both ends of $e$ matched in $T_f^{\pi}$}] \right).\]

     To prove the second part, observe that  $T_f^{\pi} \subseteq T_1^{\pi}$ 
     and $T_1^{\pi}$ is a maximal matching. For each edge of $\OPT$ that has both its end points
     unmatched in $T_f^{\pi}$, at least one end point is adjacent to an edge $T_1^{\pi}$. Since these edges
     must be part of $T_1^{\pi} \setminus T_f^{\pi}$,  
	\[
      |T_1^{\pi}| \geq |T_f^{\pi}| + \frac12 \sum_{e \in \OPT}  
            \mathbf{1}[\text{Both ends of $e$ unmatched in $T_f^{\pi}$}].            \]

\end{proof}

\noindent \textbf{Fact~\ref{swapbases}.}
\emph{Consider any matroid $\M$ and independent sets $A,B,C \in \M$  such that
    $A \subseteq \cl_{\M} (B)$ and $B \cup C \in \M$. Then we also have $A \cup
    C \in \M$. 
}
\begin{proof}
Suppose we start with $B \in M$ and add elements of $A=\{a_1, a_2, \ldots,
a_k\}$ one by the one. We show that one can ensure that the set remains
independent in $\M$ by removing some elements from $B$. First, note that $|B| =
\rank(B) = \rank(B \cup A)$. Our algorithm removes an element from $B$ only if
addition of $a_j$ creates a circuit. Hence the rank of the set is always $|B|$
and addition of every $a_j$ creates a unique circuit. Moreover, this circuit
contains an element $b_j \in B$ that can be removed as we know $A \in \M$. 

Next we repeat the above procedure but by starting with $B\cup C \in \M$ and
adding elements of $A$. We know from before that addition of each element $a_j$
creates a unique circuit that does not contain an element of $C$. Hence we can
remove element $b_j$ while ensuring the set remains independent in $\M$. This
will finally give $A \cup C \in \M$.  
\end{proof}

\else
\fi



\section{Hastiness Lemma} \label{sec:lemhastpf}
The proof of the following lemma is similar to Lemma~$2$ in~\cite{KMM-APPROX12}.

\noindent \textbf{Lemma~\ref{lemma:KMM}} (Hastiness Lemma). 
\emph{For any two matroids $\M_1$ and $\M_2$ on the same ground set $E$,  let
  $T_f^{\pi}$ denote the set selected by \Greedy after running for the first
  $f$ fraction of elements $E$ appearing in order $\pi$. Also, for $i\in
  \{1,2\}$, let $\Phi_i(T_f^{\pi}) := \cl_i(T_f^{\pi}) \cap \OPT$.     Now for
  any $0 < f, \epsilon \leq \frac12$, if $\E_{\pi}[|T_1^{\pi}|]  \leq
  (\frac{1}{2}+ \epsilon)\,|\OPT|$ then 
    \begin{align}
     & \E_{\pi} \left[ |\Phi_1(T_f^{\pi}) \cap \Phi_2(T_f^{\pi})| \right]  
            \leq 2\epsilon\,|\OPT| \qquad \text{and} \label{eq:hastin}\\
     & \E_{\pi} \left[ |\Phi_1(T_f^{\pi}) \cup \Phi_2(T_f^{\pi})| \right]  
            \geq \left(1 -  \frac{2 \epsilon}{f} + 2 \epsilon \right)\,|\OPT|. \label{eq:hastun}  
      \end{align}
    This implies  $\G(f) := \frac{\E_{\pi}[|T_f^{\pi}|] }{|\OPT|} \geq \left(\frac{1}{2} - 
                \left( \frac{1}{f} -2 \right)\epsilon \right)$.       
}

\begin{proof} 
For ease of notation, we write $T_f^{\pi}$ by $T_f$. To prove Eq.~(\ref{eq:hastin}),
    \begin{align*}
    \E_{\pi} \left[ |\Phi_1(T_f) \cap \Phi_2(T_f)| \right] &\leq \E_{\pi} \left[ |\Phi_1(T_1) \cap \Phi_2(T_1)| \right] \tag{because $T_f \subseteq T_1$} \\
    & =  \E_{\pi} \left[  \left( |\Phi_1(T_1)|  + |\Phi_2(T_1)| \right) - |\Phi_1(T_1) \cup \Phi_2(T_1)| \right] \\
    & = \E_{\pi}\left[ |\Phi_1(T_1)|  + |\Phi_2(T_1)| - |\OPT| \right]  \tag{because $T_1$ is a maximal solution}\\
    & \leq 2~\E_{\pi} \left[ |T_1|\right] - |\OPT|  \tag{because $|T_1| \geq |\Phi_i(T_1)| $} \\
    & \leq 2\epsilon\,|\OPT|.
    \end{align*}

Now to prove Eq~(\ref{eq:hastun}), we first bound $|\Phi_1(T_f)| + |\Phi_2(T_f)|$. It is at least
    \begin{align*} 
    &\quad  |\OPT|  + \sum_{e\in \OPT} \mathbf{1}[e\in \cl_1(T_f) \cap \cl_2(T_f)] - \sum_{e\in \OPT} \mathbf{1}[ e \notin \left(\cl_1(T_f) \cup \cl_2(T_f)\right)] \\
        &\geq  |\OPT|  + \sum_{e\in \OPT} \mathbf{1}[e\in T_f] - \sum_{e\in \OPT} \mathbf{1}[ e \notin \cl_1(T_f) \cup \cl_2(T_f)]. \tag{because $T_f \subseteq \cl_i(T_f)$}
        \end{align*}
Taking expectations and using Claim~\ref{claim:KMM2-Matroid},  
       \begin{align}
	 \E_{\pi}[ |\Phi_1(T_f)| + |\Phi_2(T_f)|]  &\geq |\OPT| - \Big(\frac{1}{f}-  2 \Big) \, \E_{\pi}[ |T_f \cap \OPT|] \label{eq:hastuntemp}
	 \end{align}
        { Since $f \leq \frac{1}{2}$, we can use an upper bound on $\E_{\pi}[|T_f \cap \OPT|]$. Observe $T_1 \supseteq T_f$ is a maximal solution implying $|T_1| \geq |T_1 \cap \OPT| + \frac12 ( |\OPT| -  |T_1 \cap \OPT|) \geq \frac12 ( |\OPT| +  |T_f \cap \OPT|)$. Taking expectations,  
        }
        \begin{align*}
        \E_{\pi}[|T_f \cap \OPT|] \leq 2\, \E_{\pi}
        \left[ |T_1| - \frac12 |\OPT| \right] 
         \leq 2\,\epsilon \, |\OPT|. \tag{because $\E_{\pi}[ |T_1|] \leq \left(\frac12 +\epsilon \right) |\OPT|$ }
    \end{align*}
  Combining this with Eq.~(\ref{eq:hastuntemp}) and Eq.~(\ref{eq:hastin}) proves Eq.~(\ref{eq:hastun}),
  \begin{align*} 
  \E_{\pi} \left[ |\Phi_1(T_f^{\pi}) \cup \Phi_2(T_f^{\pi})| \right]  &= \E_{\pi}[ |\Phi_1(T_f)| + |\Phi_2(T_f)|]  - \E_{\pi} \left[ |\Phi_1(T_f^{\pi}) \cap \Phi_2(T_f^{\pi})| \right]  \\
            &\geq \left(1 -  \frac{2 \epsilon}{f} + 2 \epsilon \right)\,|\OPT|.
  \end{align*}
  
  Finally, using Eq.~(\ref{eq:hastuntemp}) and  $|T_f| \geq |\Phi_i(T_f) |$, we also have $\E_{\pi}[|T_f^{\pi}|]   \geq \frac12 \E_{\pi}[ |\Phi_1(T_f)| + |\Phi_2(T_f)|]  \geq \left(\frac{1}{2} - 
                \left( \frac{1}{f} -2 \right)\epsilon \right)\,|\OPT|$.
\end{proof} 

For intuition, imagine the following claim for $f=\frac12$, where it says that for a uniformly random order probability that   $e$ is in not in the span of $T_f$ for either of the matroids is at most the probability $e$ is selected by \Greedy into $T_f$.
\begin{myclaim} \label{claim:KMM2-Matroid} Suppose  $\G(1) \leq
    \left(\frac12 + \epsilon \right)|\OPT|$ for some $\epsilon<\frac12$ and $T_f$ is the
    output of \Greedy on $E([1,mf)]$, then 
    \[ \forall e \in \OPT \qquad \Pr_{\pi}[ e \notin  \Phi_1(T_f) \wedge e \notin \Phi_2(T_f)] \leq \Big(\frac{1}{f} - 1 \Big) \Pr_{\pi}[ e  \in T_f]. \] 
\end{myclaim} 
   
\begin{proof} Let us define the event $\mathcal{X} = \Big(
    e \notin \Phi_1(T_f) \wedge e \notin \Phi_2(T_f) \Big) \vee (e \in T_f)$.  Consider
    the mapping $g$ from permutations to permutations.  If $e$ occurs in the first $f$
    fraction of elements then $g(\pi) = \pi$. If not, then remove $e$ and insert it uniformly at randomly at a
    position in $[1,mf]$.  This induces a mapping from the set of all permutations on the
    ground elements to the set of permutations that have $e $ in the first $f$ fraction of
    elements. The important observation is that the set of permutations satisfying the
    event $\mathcal{X}$ still satisfy the event under the mapping $g$. We can conclude
    that $\Pr[\mathcal{X}] \leq \Pr[\mathcal{X} \mid e \in [1,mf]] $. Conditioned on the
    event that $e \in [1, fm]$, event $\mathcal{X}$ means $e\in T_f$. This is because
    if $e \notin \Phi_1(T_f) \wedge e \notin \Phi_2(T_f)$ and $e \in E[1, fm]$
 then $T_f \cup {e} \in \M_1 \cap \M_2$.  Thus,  we can conclude that
    $\Pr[\mathcal{X}] \leq \Pr[ e \in T_f \mid e \in [1,mf]]= \frac{1}{f} \Pr[e \in T_f]
    $. Moreover, since $\Big(e \notin \Phi_1(T_f) \wedge e \notin \Phi_2(T_f) \Big)$ and $(e \in T_f)$ are disjoint events, $\Pr[\mathcal{X}] = \Pr\left[\Big(e \notin \Phi_1(T_f) \wedge e \notin \Phi_2(T_f) \Big)\right] + \Pr[e \in T_f]$, which proves this claim.  
\end{proof}

\else
\fi


\ifFULL
\section{Generalization to the Intersection of $k$-matroids}\label{sec:kmatInter}

The \emph{online $k$-matroid intersection} problem in the random arrival model (OMI)
consists of $k\geq 2$ matroids,  $\M_i = (E,\mathcal{I}_i)$ for $i\in [k]$.
  The elements of $E$ are presented one-by-one to an
online algorithm whose goal is to construct a large common independent set.  As
the elements arrive, the algorithm must immediately and irrevocably decide
whether to pick them, while ensuring that the  set of picked elements always
form a common independent set. We assume that the algorithm knows the size of
$E$ and has access to independence oracles of the $k$ matroids for the already arrived elements.

\begin{theorem} \label{thm:matroid-k-main}
    The online $k$-matroid intersection problem in the random arrival model has a
    $\left(\frac{1}{k} + \frac{\delta''}{k^4} \right)$-competitive randomized algorithm, where $\delta''>0$ is a constant.  
\end{theorem}  

The proof largely follows the proof of Theorem~\ref{thm:matroidmain} for intersection of two matroids.
We sketch the proof of the following lemma below (and make no effort in optimizing the parameters). When combined Lemma~\ref{lem:reduc}, this proves  Theorem~\ref{thm:matroid-k-main}. 
As before, $\G(1)\, |\OPT|$ denotes the expected size of the common independent produced by the greedy algorithm.

\begin{lemma}\label{lem:kmatrmainthm}
There exists constants $\epsilon,\gamma >0$ and an online algorithm such that if $\G(1) \leq \left(\frac1k + \frac{\epsilon}{k^3} \right)$ then  the algorithm  finds an independent set of expected size at least $\left(\frac1k + \frac{\gamma}{k}\right)~|\OPT|$. 
\end{lemma}

\subsection{Hastiness Lemma}
We  need the following hastiness property for the proof of Lemma~\ref{lem:kmatrmainthm}.
\begin{lemma}[Hastiness Lemma] \label{lemma:KMMgenmatr} 
  For any $k$ matroids $\M_1, \dots ,\M_k$ on the same ground set $E$,  let
  $T_f^{\pi}$ denote the set selected by \Greedy after running for the first
  $f$ fraction of elements $E$ appearing in order $\pi$. Also, for $i\in
 [k]$, let $\Phi_i(T_f^{\pi}) := \cl_i(T_f^{\pi}) \cap \OPT$.     Now for
  any $0 < f \leq \frac{1}{k}$ and $0 \leq \epsilon < 1$ , if $\E_{\pi}[|T_1^{\pi}|]  \leq
  (\frac{1}{k}+ \frac{\epsilon}{k^3})\,|\OPT|$ then 
    \begin{align}
     & \E_{\pi} \left[ |\Phi_i(T_f^{\pi}) \cap \Phi_j(T_f^{\pi})| \right]  
            \leq \frac{2\epsilon}{k^2} \,|\OPT| \qquad \text{for all }i\neq j  \in [k] \label{eq:genhastin}\\
     & \E_{\pi} \left[ \sum_{i=1}^k |\Phi_i(T_f^{\pi})| \right]  
            \geq \left(1 -  \frac{4\epsilon}{k  f} + 4 \epsilon \right)\,|\OPT|.   \label{eq:genhastun}
      \end{align}
    Hence, $\E_{\pi}[|T_f^{\pi}|]  \geq \left(\frac{1}{k} - 
                \left( \frac{4\epsilon}{k f} -4\epsilon \right)  \frac{1}{k} \right)\,|\OPT|$.       
 \end{lemma}

\begin{proof} [Proof Overview]
    For ease of notation, we write $T_f^{\pi}$ by $T_f$. We prove
    Eq.~(\ref{eq:genhastin}) by contradiction and assume 
    \[\E_{\pi} \left[
        |\Phi_i(T_f) \cap \Phi_j(T_f)| \right]  > \frac{2\epsilon}{k^2}
    \,|\OPT| \implies \E_{\pi} \left[ |\Phi_i(T_1) \cap \Phi_j(T_1)|
    \right]  > \frac{2\epsilon}{k^2} \,|\OPT|\text{ because }T_f \subseteq T_1 \]

    Let $S =  \left( \Phi_i(T_1) \cup \Phi_j(T_1) \right)$ and note that $T_1 \cup \OPT \setminus S \in \M_i \cap \M_j$. 
     Moreover,
     \begin{align*}
         |\OPT \setminus S| &= |\OPT| - \left( |\left( \Phi_i(T_1) \cup \Phi_j(T_1) \right)| \right)\\
         &= |\OPT| -  |\Phi_i(T_1) |-  |\Phi_j(T_1) | + | \Phi_i(T_1) \cap \Phi_j(T_1) |   \\
         &\geq  |\OPT| -  2\,|T_1| + | \Phi_i(T_1) \cap \Phi_j(T_1) |
     \end{align*}
         Since, $T_1 \in \bigcap_{l=1}^k \M_l$, we remove at most
         $(k-2)|T_1|$ more elements, say $S'$, from set $\OPT \setminus S$ such
         that $\left( T_1 \cup \OPT \setminus (S\cup S') \right) \in
         \bigcap_{l=1}^k \M_l$. This gives 
    \begin{align}
      |\OPT \setminus (S\cup S')| \geq |\OPT \setminus S| - (k-2)|T_1| 
        \geq |\OPT|  + | \Phi_i(T_1) \cap \Phi_j(T_1) | - k\,|T_1|. \label{eq:hastkmat} 
     \end{align}
     However, as $T_1$ is a maximal independent set, $\OPT \setminus (S\cup S')$ is an empty set. 
     Taking expectations over $\pi$ in Eq.~(\ref{eq:hastkmat}) gives $\E_{\pi}|\OPT \setminus (S\cup S')| > |\OPT| \left( 1 - k\,\left( \frac1k + \frac{\epsilon}{k^3} \right) + \frac{2\epsilon}{k^2}  \right) > 0$. This is a contradiction.

    Next, to prove Eq.~(\ref{eq:genhastun}),
  \begin{align*}
      \E[\sum_{i=1}^k |\Phi_i(T^{\pi}_f)|] 
       & = \sum_{e \in \OPT} \sum_{i=1}^k i \, \Pr[ e \text{ is in the span of $T_f$ in exactly $i$ matroids }]\\
        & = |\OPT| - \sum_{e \in \OPT} \Pr[e \text{ is in the span of $T_f$ in none of the matroids}]  \\
        & \qquad + \sum_{e \in \OPT}  \sum_{i=1}^k (i-1) \, \Pr[ e \text{ is in the span of $T_f$ in exactly $i$ matroids }] \\
        &\geq  |\OPT| + (k-1)\,\sum_{e \in \OPT} \Pr[e \text{ is in the span of $T_f$ in all the matroids}] \\
         &\qquad       - \sum_{e \in \OPT} \Pr[e \text{ is in the span of $T_f$ in none of the matroids}]\\
         \intertext{Using Claim~\ref{claim:KMM2-k-Matroid} below and the fact that 
    $\Pr[ e \text{ is in the span of $T_f$ in all the matroids}] \geq \Pr[e \in T_f]$ , we have}
          &\geq |\OPT| - \left( \frac{1}{f} - k \right) \sum_{e \in \OPT}\Pr_{\pi}[e \in T_f] \\
          &= |\OPT| - \left(\frac{1}{f} -k \right)\E[ |\OPT \cap T_f| ]
      \end{align*}
      Since $f \leq \frac1k$, we use an upper bound
      $\E[ |\OPT \cap T_f|] \leq \E[|\OPT \cap T_1|] \leq  \frac{2\epsilon}{k-1} \leq \frac{4\epsilon}{k}$ to finish
      the proof. 
      To prove this bound, observe $T_1$ is a maximal set, and achieves a $k$ approximation. We can conclude that
      $|T_1| \geq |\OPT \cap T_1| + \frac{1}{k} (|\OPT| - |\OPT \cap T_1|)$.  Taking expectations
      and simplifying we get $\E[|\OPT \cap T_1|] \leq \frac{2 \epsilon}{k-1}$.

Since for all $i$, we have $|T_f^{\pi}| \geq |\Phi_i(T_f^{\pi})| $, the hence part of the lemma follows because $|T_f^{\pi}| \geq \frac1k \sum_i |\Phi_i(T_f^{\pi})| $.
Finally,  the proof of the following claim  is similar to that of Claim~\ref{claim:KMM2-Matroid}.
\begin{myclaim} \label{claim:KMM2-k-Matroid} Suppose we know that $\G(1) \leq
    \left(\frac1k + \frac{\epsilon}{k^2} \right)|\OPT|$ for some $\epsilon>0$ and $T_f$ is the
    output of \Greedy on the $E([1,mf)]$, then 
    \[ \forall e \in \OPT \qquad \Pr_{\pi}[ e \notin \bigcup_{i=1}^k  \Phi_i(T_f) ] \leq \Big(\frac{1}{f} - 1 \Big) \Pr_{\pi}[ e  \in T_f] \] 
\end{myclaim} 
\end{proof} 



\subsection{Modifications to the \MkGreedy Algorithm}

 Phase~(a)  of the algorithm remains the same; we will 
 use the first $f$ fraction to pick $1-p$ fraction of the elements chosen by \Greedy. Let $T_f$ 
 denote the elements chosen by \Greedy and $S$ to be the elements picked into the final solution. 

 In Phase~(b), the algorithm is modified in a natural way; we pick elements that lie 
 in $\M_1/T \cap \M_2/T \cap \dots \cap \M_i/S \cap \dots \cap \M_k /T$ for each $i \in \{1,\dots,k\}.$ 
 When $k=2$, this reduces to the algorithm given in Section~\ref{section:matroids}. Let $N_i$
 denote the set of elements chosen by \Greedy on the matroids 
 $\Greedy(\M_1/T_f, \dots,\M_i/S, \dots,\M_k/T_f)$. We will return
 $S\cup \big(\bigcup_{i=1}^k N_i\big).$ Using Fact~\ref{swapbases}, we know that 
 $S \cup \bigcup_{i=1}^k N_i$ is an independent set in all the matroids
 $\M_1,\dots,\M_k$. To show that $\E[|\bigcup_{i=1}^k N_i|] $ is large, we use the following  modified sampling lemma.

\begin{lemma}~(Sampling Lemma for $k$-matroids). 
\label{lem:k-matrmain} 
    Given matroids $\M_1, \M_2,\dots \M_k$ on ground set $E$, a set $T \in \bigcap_{j=1}^k \M_j$, 
    and $\K(e) \sim  \text{Bern}(1-p)$ i.i.d. for all $e \in T$, we define set 
    $S := \{e \mid e \in T \text{ and } \K(e) =    1\}.$ I.e., $S$ is a set achieved by dropping 
    each element in $T$ independently with probability $p$.
   For $i\in [k]$,  consider a set $\tE \subseteq \cl_i(T)$ and a set
   ${\tI} \subseteq \tE $ satisfying  ${\tI} \in \Mi \cap (\bigcap_{j\neq i} \M_j/T) $. Then 
    for any arrival order of the elements of $\tE$, we have
    \[ \E_{\K}[ \text{\Greedy}(\M_1/T, \dots, \Mi/S,\dots, \M_k/T, \tE)] \geq \frac{p}{1+p(k-1)} |{\tI} | .\] 
\end{lemma}
\begin{proof}[Proof Overview]
    We use the same proof outline as Lemma~\ref{lem:matrmain} and prove for $i=1$ as the other cases are analogous. Once again, the performance of \Greedy can be mapped to a Sampling Algorithm like \SampAlg and we analyze its performance. The only
    difference is in Line~\ref{SampAlg:M2check} where we now ask that the new element $e$ is independent in
    all the matroids $ \M_j$ for all $j\neq 1$ (instead of just $\M_2$).  Observations~\ref{clm:SAinvariant2} and \ref{clm:SAinvariant1}, and Claim~\ref{clm:SAcircuits}
    remain the same for the new Sampling Algorithm. 

    We first show that at the end of  each iteration we can maintain an invariant. 
    \begin{invariant}
For given sets $S', N', T$, and $\trE \subseteq \tE$, we have set $\I'$ satisfying this invariant if 
\begin{align*}
S' \cup N' \cup I' &\in \M_1     \\ 
            T \cup N' \cup I' &\in \M_j \qquad \text{for $j \in [2,k]$}   \\
            I' &\subseteq \trE 
\end{align*}
\end{invariant}

   This invariant still contains Eq.~(\ref{inv1}) and Eq.~(\ref{inv3}) of Invariant~\ref{invariant}; however, it  contains one equation $(T \cup N' \cup \I') \in \M_j$ for each $j \in [2,k]$ (instead of just $\M_2$). The updates are also naturally
    extended. Whenever adding a new element violates any of these invariants, 
    we simply remove some elements     from $I'$ to compensate. 
    \begin{updates} We perform the following updates to $\I'$
    whenever \SampAlg reaches Line~\ref{SampAlg:Sadd} or
    Line~\ref{SampAlg:Nadd}. 
    
\begin{itemize}
\item Line~\ref{SampAlg:Sadd}: If circuit $C_1(S' \cup N' \cup \I' \cup f)$ is
    non-empty then remove an element from $\I'$ belonging to $C_1(S' \cup N'
    \cup \I' \cup f)$ to break the circuit. 

\item Line~\ref{SampAlg:Nadd}: If circuit $C_1(S' \cup N' \cup \I' \cup e)$ is
    non-empty then remove an element from $\I'$ belonging to 
    $C_1(S' \cup N' \cup \I' \cup e)$ to break the circuit.  For $j\in [2,k]$, if 
    $C_j(T \cup N' \cup \I' \cup e)$ is non-empty then remove another element from $\I'$ belonging
    to $C_j(T \cup N' \cup \I' \cup e)$ to break the circuit. In the special
    case where $e \in \I'$, we remove $e$ from $\I'$.

\end{itemize}
\end{updates}
A claim similar to Claim~\ref{claim:invariant} shows that the above updates are well-defined and maintain the invariants.
    Now using the invariants, we prove that the expected number of elements picked is large. As before, we apply the principle of 
    deferred decisions and define the events $B_j$ and $\thickbar{B}$ for $j\in [0,l-1]$, where $l = |C|$. Let $\alpha:= \frac{p}{1+p(k-1)}$. To prove the induction step, in the event $B_j$, we use I.H. to get $\E_{\K}[|N'| \mid B_j ] \geq 1+ \alpha \, (|I'| -j -k )$ and in the event $\thickbar{B}$, we use I.H. to get $\E_{\K}[|N'| \mid \thickbar{B} ] \geq  \alpha \, (|\I'| - l)$. A little algebra completes the proof.
\end{proof}

\subsection{Putting Everything Together}
\begin{definition} [Sets $\tE_i$]
     For $i \in [k]$, we define $\tE_i$ 
    to be the set of elements $e$ that arrive in Phase~(b), $e \in \cl_{i}(T_f)$, and $e
    \not\in \cl_{j}(T_f)$ for $j \neq i$. 
\end{definition}

To prove Lemma~\ref{lem:kmatrmainthm} using the Sampling Lemma~\ref{lem:k-matrmain}, we show that there exist ``large'' disjoint
subsets of $\OPT$ that are still to arrive in Phase~(b).     For $i \in [k]$, let 
$\Phi_i(T_f) := \cl_i(T_f) \cap \OPT$ and  
$I_i := \Phi_i(T_f) \setminus \big( \bigcup_{j\neq i} \Phi_j(T_f) \big)$.

\begin{myclaim} \label{claim:kpropSetI}
 If $\G(1) \leq \left(\frac1k + \frac{\epsilon}{k^3} \right)$ then for $i\in[k]$ there 
 exist disjoint sets $\tI_i \subseteq \tE_i$ such that
   \begin{enumerate}[(i)]
        \item $\E_{\pi}\big[ \sum_i|{\tI}_i| \big] \geq   \left(1 + 4\epsilon  (k - \frac{1}{f}- \frac{1}{2} )-  \frac{\epsilon}{k} + \frac{\epsilon}{k^2} - f \right) |\OPT| $  \label{eq:kclaimtildeIsize}
       \item ${\tI}_i \in \M_i \cap \left( \bigcap_{j\neq i}\M_j/T_f \right)$.  \label{eq:kclaimtildeI}
    \end{enumerate}
\end{myclaim}

\begin{proof}[Proof Overview]
    For any $i \in [k]$, we first construct $\I_i'' \subseteq I_i$ satisfying
     Eq.~(\ref{eq:kclaimtildeI}). 
  For a fixed $i$  and each $j \neq i$, let $X^i_j \subseteq T_f$ be a set of minimum size that ensures
    $\cl_{j}(X^i_j \cup \Phi_j(T_f))   = \cl_j(T_f)$.     Here $|X^i_j|  = |T_f| - |\Phi_j(T_f)|$ because $\M_j$ is a matroid.
    Since    $I_i \cup \Phi_j(T_f)   \in \M_j$ and $I_i \cap \Phi_j(T_f) = \emptyset$, by 
    eliminating at most $|X^i_j|$ elements from $\I_i$, we can create $\I_i' \subseteq \I_i$ 
    that satisfies 
    $X^i_j \cup \I_i' \cup \Phi_j(T_f) \in \M_j$.  
    Now using using Fact~\ref{swapbases} and  $\cl_j(X^i_j \cup \Phi_j(T_f) ) = \cl_j(T_f)$, we  conclude  $\I_i' \in \M_j/T_f$.   Hence, we can  define 
    $\I_i'' := \I_i \setminus (\bigcup_{j\neq i} X^i_j)$ where $\I_i'' \in \M_j /T_f$ for all $j\neq i$.
    Observe, 
    \[ |\bigcup_{j\neq i} X_j^i | \leq \sum_{j\neq i}|X^i_j| = \sum_{j\neq i} \big( |T_f| - |\Phi_j(T_f)|\big) \]
    Combining this with  $|I_i| \geq |\Phi_i(T_f)| - \sum_{j\neq i} |\Phi_j(T_f) \cap  \Phi_i(T_f)|$ (a fact from the definition of $\I_i$), we get 
    \begin{align*}
     |\I_i''| = |\I_i| - |\bigcup_{j\neq i} X_j^i| &\geq |\Phi_i(T_f)| - (k-1) |T_f|+  \sum_{j\neq i} (|\Phi_j(T_f)| -  |\Phi_j(T_f) \cap  \Phi_i(T_f)|) \\
     &= \sum_j |\Phi_j(T_f)| - (k-1)|T_f| - \sum_{j\neq i} |\Phi_j(T_f) \cap  \Phi_i(T_f)|.
     \end{align*}
    Note that for all $i$ sets $I_i''$ are pairwise disjoint (as $I_i$ are constructed to be disjoint). Furthermore,
    taking expectations over $\pi$, we let $\tI_i$ to be the elements of $\I''$ that are
    still to appear in Phase~(b). Hence, $\E_{\pi}[ \sum_{i=1}^k |\tI_i|] \geq \E_{\pi}[ \sum_{i=1}^k |I_i''|] - f|\OPT|$. 
    Using the bounds from Hastiness Lemma~\ref{lemma:KMMgenmatr}, we conclude
    \begin{align*} 
        &~~\E_{\pi} \left[ \sum_{i=1}^k |\tI_i| \right] \\
        &\geq \sum_{i=1}^k \Big( \sum_j [|\Phi_j(T_f)|] - 
                     (k-1)\, \E_{\pi}[|T_f|] -  \sum_{j\neq i}  \E_{\pi}[ |\Phi_i(T_f) \cap \Phi_j(T_f)|]\Big) - f|\OPT|  \\
       &\geq k \left(1- \frac{4\epsilon}{kf}  + 4\epsilon-  (k-1)\,\G(1)  - k \cdot \frac{2\epsilon}{k^2} \right) |\OPT|  - f |\OPT| \\
       & \geq 	\left( k+ 4 \epsilon ( k-  \frac{1}{ f}  - \frac{1}{2} )  - k(k-1) (\frac1k + \frac{\epsilon}{k^3}) - f   \right) |\OPT| \\
       &  = \left( 1+ 4 \epsilon ( k-  \frac{1}{ f}  - \frac{1}{2} ) - \frac{\epsilon}{k} + \frac{\epsilon}{k^2}  - f   \right) |\OPT|.  
    \end{align*}
\end{proof}

Finally, to prove Lemma~\ref{lem:kmatrmainthm}, we use the disjoint sets $\tI_i$ from Claim~\ref{claim:kpropSetI} in the Sampling Lemma~\ref{lem:k-matrmain} to say that the expected output size is at least 
\[ (1-p) \G(f)  + \frac{ p}{1+p(k-1)} \left(1 + 4\epsilon  (k - \frac{1}{f}- \frac{1}{2} )-  \frac{\epsilon}{k} + \frac{\epsilon}{k^2} - f \right) |\OPT|  
\]

We assume $k\geq 3$ (as the case $k=2$ was presented in Section~\ref{section:matroids}). By using Lemma~\ref{lemma:KMMgenmatr}, setting $f= \frac{1}{k}$, and choosing
$\epsilon \ll 1$, we can conclude that the expected value is at least
\[ (1-p) \left( \frac{1}{k} - \frac{4\epsilon  - 4 \epsilon}{k} \right)+ \frac{p}{1+p(k-1)} \Big( \frac{k-1}{k} - 3\epsilon \Big). \]
This value is at least $\frac{1}{k} + \frac{\gamma}{k}$ for some universal
constant $\gamma > 0$ (e.g., when $p=0.2$ and $\epsilon = 10^{-10}$). 

\else
\fi

\end{document}